\documentclass[12pt]{extarticle}
\usepackage{amsmath}
\usepackage{amsfonts}
\usepackage{color}      
\usepackage{mathenv}

\usepackage{amsthm}

   \newtheorem{theorem}{Theorem}[section]
   \newtheorem{lemma}[theorem]{Lemma}
   \newtheorem{proposition}[theorem]{Proposition}

   \newtheorem*{condition}{Condition}

   \newenvironment{definition}[1][Definition]{\begin{trivlist}
   \item[\hskip \labelsep {\bfseries #1}]}{\end{trivlist}}
   \newenvironment{example}[1][Example]{\begin{trivlist}
   \item[\hskip \labelsep {\bfseries #1}]}{\end{trivlist}}
   \newenvironment{remark}[1][Remark]{\begin{trivlist}
   \item[\hskip \labelsep {\bfseries #1}]}{\end{trivlist}}

   \numberwithin{equation}{section}

\renewcommand{\thefootnote}{\arabic{footnote}}

 \newcommand{\tn}{\textnormal}

\newcommand{\zee}{\mathbb{Z}}       
\newcommand{\zd}{\mathbb{Z}^d}       
\newcommand{\edge}{\mathbb{E}}       
\newcommand{\lattice}{\mathbb{L}}    
\newcommand{\lxt}{\mathbb{Z}^d\times\mathbb{R}}   

\newcommand{\be}{\begin{equation}}  
\newcommand{\ee}{\end{equation}}   
\newcommand{\beq}{\begin{eqnarray}}  
\newcommand{\eeq}{\end{eqnarray}}

\newcommand{\cA}{\mathcal{A}}  
  
\newcommand{\cC}{\mathcal{C}}   
   
\newcommand{\cF}{\mathcal{F}}   
\newcommand{\cH}{\mathcal{H}}   
\newcommand{\cK}{\mathcal{K}}   
   
\newcommand{\cM}{\mathcal{M}}  
\newcommand{\cO}{\mathcal{O}}   
\newcommand{\cS}{\mathcal{S}}   
\newcommand{\cU}{\mathcal{U}}

\newcommand{\bbC}{\mathbb{C}}   
\newcommand{\bbE}{\mathbb{E}}   
\newcommand{\bbN}{\mathbb{N}}   
\newcommand{\bbP}{\mathbb{P}}  
\newcommand{\bbQ}{\mathbb{Q}}   
\newcommand{\bbR}{\mathbb{R}}   
\newcommand{\bbT}{\mathbb{T}}   
\newcommand{\bbZ}{\mathbb{Z}}

\newcommand{\bfb}{{\bf b}}   
\newcommand{\bff}{{\bf f}}   
\newcommand{\bfp}{{\bf p}}    
\newcommand{\bfw}{{\bf w}}    
\newcommand{\bfx}{{\bf x}}    
\newcommand{\bfy}{{\bf y}}

\newcommand{\bfB}{{\bf B}}   
\newcommand{\bfC}{{\bf C}}    
\newcommand{\bfD}{{\bf D}}     
           
\newcommand{\bfQ}{{\bf Q}}           
\newcommand{\bfT}{{\bf T}} 

\newcommand{\bfOm}{{\bf \Omega}}  
\newcommand{\bfTh}{{\bf \Theta}}

\newcommand{\tql}{\textquoteleft}   
\newcommand{\tqr}{\textquoteright}   

\newcommand{\real}{\mathbb{R}} 

\newcommand{\cons}{{\bf \Omega}}   

\newcommand{\csm}[2]{{\bf Q}_{#1;\lambda}^{#2}} 

\newcommand{\con}{\omega } 

\newcommand{\lra}{\leftrightarrow}   
\newcommand{\llra}{\longleftrightarrow}   

\newcommand{\less}{\backslash}                                

\newcommand{\ag}{{\bf T}}                                     

\newcommand{\spc}{{\bf \Theta}}                               

\newcommand{\tr}{\mathrm{tr}}

\newcommand{\wt}{\widetilde}
\newcommand{\wh}{\widehat}
\newcommand{\ol}{\overline}

\begin{document}

\title{Localization for the Ising model in a transverse field with generic aperiodic disorder}
\author{Rajinder Mavi}
\maketitle
{\let\thefootnote\relax\footnote{This work was supported by the Institute of Mathematical Physics at Michigan State University and NSF Grant DMS-1101578} }
\begin{abstract} 
  We show that the transverse field Ising model undergoes a zero temperature phase transition for a $G_\delta$ set of ergodic transverse fields. We apply our results to the special case of quasiperiodic  transverse fields, in one dimension we find a sharp condition for the existence of a phase transition.
\end{abstract}

\section{Introduction}

   The Hamiltonian of the transverse field  quantum Ising model on $\Lambda \subset \subset \bbZ^d$  is defined as,
    \begin{equation}   \label{imtf}
         H_\Lambda 
         = - \sum_{{x,y\in \Lambda}:{\|x-y\| = 1}} 
          \frac{ \lambda }{2}\sigma_x^{(3)}\sigma_y^{(3)}     -  \sum_{x\in \Lambda}  \delta(x) \sigma_x^{(1)}
     \end{equation} 
   where, the $\sigma^{(i)}$ are the Pauli spin matrices, ie for $x\in \Lambda$,
 \[     \sigma^{(3)}_x = \left( \begin{matrix} 1 & 0 \\ 0 & -1 \end{matrix}  \right),  \hspace{1in}  
        \sigma^{(1)}_x = \left( \begin{matrix}  0 & 1 \\ 1 & 0 \end{matrix}  \right ).   \]
     $H_\Lambda$ acts in the Hilbert space $\cH_\Lambda = \otimes_{x\in \Lambda} \bbC^2 $. 
    We will consider the behavior of correlations $\langle \sigma^{(3)}_x \sigma^{(3)}_y \rangle_{\delta,\lambda}$ in the ground state.

    For $\delta(x)$ i.i.d. random variables on $(0,1)$ and low  interaction  $\lambda \searrow 0$ between spins, the ground state is known to exhibit both long range order and short range ordered phases depending  on the behavior of the distribution function $g (s) = \bbP(\delta < s) $ as $s \to 0$. The long and short range ordered phases are formally defined by the behavior of the  order parameter (\ref{order}) which is defined below. 
    
     To explain the appearance of either phase in the ground state in the  small $\lambda$ regime it is helpful to first recall the behavior for constant $\delta$. For any dimension, there is a critical ratio $\rho_c > 0 $, so that for constant $\delta$ obeying $0 < \delta/\lambda < \rho_c$, the ground state attains long range order with spontaneous magnetization. On the other hand, for constant $\delta$ satisfying $ \delta/\lambda > \rho_c$ the ground state has short range order characterized by exponential decay of spin correlations. The sharpness of the phase transition  was proven concurrently in \cite{BG09}, \cite{CI10}.
     
    For the non-constant case, for some choices of parameters, an ordering principle can determine the phase of the ground state. For  fields $\delta$ so that $ 0 <\delta(x) \leq 1 $ for all $x$, fixing $\lambda > \rho_c^{-1} \geq \rho_c^{-1}\sup \delta(x)$  implies that the ground state is in the long range order phase. If there is some $s > 0$ so that $\delta(x) > s $ for all $x$ (which holds almost surely if $g(s) = 0$), then fixing $\lambda < s/\rho_c$ implies exponential decay of spin correlations in the ground state. 

       The ordering principle does not apply for small $\lambda$ if $g (s) > 0$ for all $s > 0$. In that case, one must resort to alternative arguments taking into account the geometry of the  regions $\Lambda \subset\bbZ^d$ where $\min_{x \in \Lambda} \delta(x)< \rho_c \lambda$, to determine the phase of the ground state.
       
     If the distribution function $g$ defining the random field $\delta$ is sufficiently nice, there is a disordered phase for small enough $\lambda$. It is known that there exists an $\alpha_d$,  so that, if $\alpha > \alpha_d$ and $\limsup_{s\to \infty} s^\alpha g(e^{-s}) < \infty$ then there is $\lambda_\alpha$ so that $ 0 < \lambda < \lambda_\alpha$ implies the system is in the short range order phase. The system is in the long range ordered phase for $\lambda > \rho_c^{-1} $ if, for all $x$, $0<\delta(x) \leq 1$. Thus, the phase transition of long range order for large $\lambda$ to short range order for small $\lambda$ is preserved. Inasmuch as the random field preserves a phase transition at positive $\lambda$ and therefore resembles the constant $\delta$ system, this regime is known as the  {\it weak} disordered field. Notice, if the density $\tfrac{dg}{ds} $ is bounded, the field is weakly disordered.

     On the other hand, if the random field is such that 
 $\liminf_{s\to \infty}s^{d}g(e^{-s}) = \infty$, the locally  correlated regions  will percolate for any $\lambda > 0$. Therefore, in this regime the ground state is  long range ordered for all $\lambda > 0$. As this disorder regime eliminates a phase transition at positive $\lambda$ it is known as {\it strong} disorder. See Section \ref{randomcon} for details of the random disordered field case.

     Although it is known that the weakly disordered field obtains  a short and a long range ordered phase, at large and small $\lambda $ respectively, the details of the phase transition are undetermined. In particular, though it is known in the clean system $\delta(x) \equiv \delta$, that the phase transition  is sharp \cite{BG09}, \cite{CI10}   no corresponding result is known for disordered systems.

      As an alternative to characterizing the (almost sure) phase diagram of spin   models with    disordered defects
     one may consider the phase diagram for (a generic class of) ordered defects.

       In this paper we will consider dynamically defined transverse fields $\delta$, which include, for example, quasiperiodically ordered fields.  We show  that the ground state transition from a long range ordered phase for  $\lambda > \rho_c^{-1}$ to  a short range order phase as   $\lambda \searrow 0$   persists for topologically generic, ie dense $G_\delta$,   ordered defects of the transverse field. Note that this corresponds to the random case with weak disorder.
      As with the disordered case we find there are atypical, yet dense,  ordered $\delta$ for which the ground state is in a long range ordered phase for any $\lambda > 0$.
      We discuss the topology of the  dynamically defined fields below,
     see Section \ref{main} for the metric defining the topology of sampling functions given a dynamical system.

    Concretely, we define long range order as the presence of sponteous magnetization and short range order as the absence of spontaneous magnetization.

      Absence of spontaneous magnetization is  a relatively  weak indicator of localization. 
        One may ask for stronger indicators of localization such as exponential clustering of correlations or exponential decay of entanglement.
        Indeed we appeal to a multiscale argument to obtain exponential decay of correlations in a dense set of dynamically ordered environments.
        The multiscale argument is similar to the approach developed for the disordered model  \cite{cam91},\cite{ckp91} and quasiperiodic models in \cite{jito}.  
      We note that   exponential decay of entanglement
      has  been demonstrated for
      the ground state of (\ref{imtf})   for small $\lambda$ in the weak disorder regime \cite{gos08},
      the results of that paper rely on the multiscale method developed in \cite{cam91},\cite{ckp91} 
      and should carry into environments with  dynamically ordered defects we consider in this paper.

     In terms of the model parameters, the transverse field Ising model (\ref{imtf}) is a boundary case of the anisotropic XY model in a transverse field. In the one dimensional setting, it is well known that under the Jordan-Wigner transformation, the one-dimensional $XY$ model may be transformed to a system of non-interacting fermions.
   This transformation carries over to the transverse field Ising model as well, in this case,     at the ground state, the single particle Hamiltonians reduce simply to Jacobian matrices defined by nearest neighbor hoppings $H_{2i,2i+1} = \delta(i)$ and $H_{2i+1, 2i+2} = 1$. 
     Chapman and Stolz \cite{CS2015} utilize this construction to investigate the spectrum, including the ground state of (\ref{imtf}) in the case of random $\delta$ so that $|E \log \delta| < \infty $. However, this leaves the nature of the ground state in the case $E|\log \delta| = \infty$ undetermined, which is precisely the case of interest in the one dimensional setting \cite{AizKleinNewman}, (cf. also, the discussion in Section \ref{randomcon}).  On the other hand, the behavior of the single particle systems for quasiperiodic potentials is well understood only for analytic sampling functions, so it does not extend to the current context.    We emphasize that  the above results pertain to the one dimensional model whereas 
       the results in this paper address general d - dimensional models.

   \subsection{The Ising model} 
    The relevant family of two point functions of $H_\Lambda$  in finite volumes $\Lambda \subset \subset \bbZ^d $ and temperatures $ 0 < \beta < \infty$ are defined as 
     \be \label{twopoint}   
             \langle \sigma_x^{(3)} \sigma_y^{(3)} \rangle^{(\Lambda,\beta)}_{\delta,\lambda}  : = 
             \frac{  \tr( \sigma_x^{(3)} \sigma_y^{(3)} e^{-\beta H_\Lambda}   )  } { \tr ( e^{-\beta H_\Lambda}) } \ee
    where we take free boundary conditions on finite sets $\Lambda\subset\subset \bbZ^d$. As we take $\beta \to \infty $ and $\Lambda \to \bbZ^d$, the limiting quantity is the expectation of $\sigma^{(3)}_x\sigma^{(3)}_y$ with respect to the infinite volume ground state. In Sections (\ref{clasd}) and (\ref{posy}) we will show the existence of such limits for general polynomials in $(\sigma_{x_i}^{(\alpha_i)})$, which confirms the existence of the ground state. We will denote this limit by dropping the $\Lambda$ and $\beta$ from the notation.    
        
    The behavior of (\ref{twopoint}) will determine the phase of the ground state  of (\ref{imtf}). We define the order parameter at $x$ as  
         \be \label{order}
        M_{\delta,\lambda}(x):=  \lim_{L \to \infty} \sup \{  \langle \sigma_x^{(3)} \sigma_y^{(3)} \rangle_{\delta,\lambda} : \|y-x\| > L \}.  
         \ee 
     We say the ground state is in the long range order phase if, for each $x\in \bbZ^d$,  $M_x > 0$. 
        On the other hand, the ground state is in the short range order phase if $M_x = 0$. 
  In fact, we will see from the FK representation below that it is sufficient to determine the value of $M_{\delta,\lambda}$ at 0 to determine if the ground state is long or short range ordered.

   \subsubsection{Description of the model and main results} \label{main}

    We will define the transverse field $\delta$ by a sampling function over a dynamic system on a 
     compact metric space $(\bfTh, r)$.
  Let $\bfT$ be a group action $\bfT:(\bfTh, \bbZ^d ) \to \bfTh$ 
  defined by a set of continuous commuting automorphisms,  
       $\{T_i\}_{i=1}^d$. 
    We will write, for $x \in \bbZ^d$ and $\theta \in \bfTh$,
   $\bfT: (\theta,x) \mapsto \bfT^x\theta := T_1^{x_1} \cdots  T_d^{x_d} \theta$.
We require the group generated by  $\left\{T_i\right\}_{i=1}^d$ to be {\em aperiodic}, 
   that is, for all $x\in \bbZ^d\setminus\{0\}$, the map $\bfT^x$ has no fixed points.
   Moreover, we require the set of automorphisms to share 
   an ergodic probability measure $\mu$ on $(\bfTh,r)$. Let us sum up this construction in the following definition.
   \begin{definition} We say $(\bfTh,\bfT  )$ is an environment process if $\bfT$ is a $\bbZ^d$ group action on the compact measure space $(\bfTh,\mu)$ so that each  $T_i$ is  ergodic with respect to measure $\mu$.
   \end{definition}

We will consider  sampling functions which are non-negative, continuous and have a non-empty finite zero set. Let us denote this set by
\be\label{samfun}
  \cC_{fin}^+(\bfTh)   =  
   \left\{h\in \cC(\bfTh): 
   \forall \theta\in \bfTh,\  h(\theta) \geq 0;
    1 \leq |h^{-1}(0)| < \infty \right\}
\ee 
 to which we associate the usual   $L_\infty$ metric for continous functions: $d_{\infty}(h,h') = \|h-h'\|_\infty$.     For initial condition $\theta$, and $h \in \cC^+_{fin}(\bfTh)$, define the transverse field by $\delta(x) = h(\bfT^x \theta)$. 
 
 \begin{example}
 
 The construction above generalizes quasiperiodic sampling functions which was studied in\cite{jito}. Here, for some $n > 0$ let $ \bfTh = \bbT^n = \bbR^n /\bbZ^n $  and for any $d\geq 1 $. Let  $A\in M_{n\times d}$ be a matrix so that each column vector $A_i$ of $A$ generates an ergodic shift $(x_i,\theta) \to T_i^{x_i}\theta = \theta + x_iA$ and the set  $\{A_1,..,A_n, e_1,...,e_n\}$ is rationally independent (the vectors $e_i $ are the standard basis in $\bbR^n$). The group action is defined as  $\bfT^x\theta = Ax+\theta$ for $\theta \in \bbT^n$. Thus the general $n$ frequency $d$ dimensional quasiperiodic field is defined, for some sampling function $h $ and initial condition $\theta \in \bfTh$, by $\delta(x) =h( \bfT^x\theta)$. For $n = d = 1$ we have $ \bfTh = \bbT$ and the one frequency one dimensional quasiperiodic trasverse field is defined by $\delta(x) = h(\theta + \omega x)$ for an irrational value $\omega \in \bbR \setminus \bbQ.$ An explicit example may be constructed as follows. Let $f$ be an analytic function in a neighborhood of the unit circle in $\bbC$ such that $f$ is real and non-negative on the unit circle. We may then take a sampling function to be $h(\theta) = f(e^{i 2\pi  \theta})$.

 As a second example we introduce the skew shift in one dimension $d =1$. Here, let $\bfTh = \bbT^2$, define $ \bfT(\theta_1, \theta_2) = (\theta_1 + \omega, \theta_1 + \theta_2)$ for irrational $\omega$.    More generally, our results apply to any   minimal ergodic system on a compact set. Finally, note the construction of  aperiodic fields rules out i.i.d. random fields.
  
  \end{example}

   The two point function is invariant
   under scaling,
   $\langle \sigma^{(3)}_x 
   \sigma^{(3)}_y\rangle_{\delta,\lambda} 
   =\langle \sigma^{(3)}_x 
   \sigma^{(3)}_y\rangle_{k \delta, k \lambda} $,
   for any $k > 0$.  
   As we are interested in the $\lambda \to 0$
   limit, we control for this degree of freedom by 
   identifying a sampling 
   function with its normalization.
   Let us define $\tau(h) := h/\|h\|_\infty$
   which maps $\cC_{fin}^+(\bfTh)$ onto
   \begin{equation} \label{samfun1}
   \cC_{fin,1}^+(\bfTh) := 
     \{ h \in \cC_{fin}^+(\bfTh) : \|h\|_\infty = 1 \}, 
   \end{equation}
   which is again equipped with the $d_\infty $ metric.
   For $h,h'\in \cC_{fin}^+(\bfTh)  $, 
   let us define distance with respect
   to this normalization as
      $d_\tau(h,h') := d_\infty(\tau(h),\tau(h'))$.

    Given a initial condition $\theta$ and sampling function h,
     we define the  magnetization parameter  similar to (\ref{order}).
     Let $\delta(x) = h(\bfT^x\theta)$ and define
     $ M_{h ,\theta;\lambda}  := M_{\delta,\lambda} (x)$. 
     It follows from ergodicity and the FK representation that, for any choice of $h,\lambda$, the set $\cU_0 = \cU_0(h,\lambda) = \{\theta \in \bfTh:M_{h,\theta;\lambda}  = 0  \} $ has either full or zero measure.

    We say a sampling function $h$ admits a short range phase (corresponding to weak disorder) if there is some critical $\lambda_h$     so that for $0 < \lambda  < \lambda_h$, $\mu(\cU_0) = 1$. 
    Let us denote the set of functions admitting a short range phase by $F_{weak}$.
    We say a sampling function is short range free  if, for all $\lambda > 0$, we have $\mu (\cU_0) = 0$. Let us denote the set of short range free functions by $F_{strong}$, note the ground state defined by such sampling functions are always in the long range ordered phase.

   \begin{theorem}\label{introtheorem}
       Let $(\bfTh, \bfT)$ be an environment process.
       Then the set of functions $F_{short}$  admitting a short range phase  is a dense $G_\delta$ 
       in the topology of $( \cC_{fin}^+(\bfTh),d_\tau)$.
         Moreover,
         $\cC_{fin}^+(\bfTh)$ partitions into
          $F_{weak}\sqcup F_{strong}$ and
          $F_{long}$ is dense in $(\cC_{fin}^+(\bfTh), d_\tau)$.
   \end{theorem}
 
 \begin{remark}
  In fact, the placement of a function into $F_{weak}$ or $F_{strong}$ depends only on the behavior of the function near the zero set. Thus, once some condition for short range free sampling function is determined, density of the set $F_{strong}$ is almost immediate as only a small perturbation is required near the zero set.  
  \end{remark}

  For actions $\bfT$ which are uniquely ergodic we can slightly improve the characterization of the long range phase for sampling functions in $F_{weak}$.
 
  \begin{theorem}\label{introtheorem2}
         Let $(\bfTh,\bfT)$ be an environment process.
       If at least one $T_i$ is uniquely ergodic  
          then for any $\lambda$  we have
          $  M_{h,\lambda} \equiv 0$ or there is some $m > 0$ 
            so that uniformly $ M_{h,\lambda} (\theta) \geq m$.
  \end{theorem}
  
  \begin{remark}
  Note that Theorem \ref{introtheorem2} is not true in the random case. Indeed, in the random case for $\lambda$ such that  $g( \rho_c \lambda ) < 1$, and for any $L  < \infty$ with probability 1, there are sites $x$ so that $\delta(y )  > \rho_c \lambda $ for all $ y \in \Lambda_L(x) = \{y\in \bbZ^d:\|y-x\| < L\}$. Thus, there is a sequence $x_i$ so that $\limsup_{i} M_{\delta,\lambda}(x_i) = 0$ with probability 1.

 \end{remark}

   \subsubsection{The FK representation, conditions for $F_{strong}$ and $F_{weak}$, and application to the quasiperiodic case}
     The majority of the analysis will take place in the Fortuin Kasteleyn (FK) representation of the ground state of (\ref{imtf}).
    The FK representation in this case is a percolation model which takes place on $\bbZ^d \times \bbR$.
    We write $ \bfQ_{\delta,\lambda}^{(2)}  \left((x,t)\lra   (y,s)\right) $ for the probability that $(x,t) $ and $(y,s)$ are 
     in the same component given the transverse field $\delta $ and coupling $\lambda$.
    A complete definition of the model is given in Section \ref{models}.
    As discussed in the introduction, utilizing a multiscale analysis, we obtain a stronger form of localization than absence of spontaneous magnetization. In fact, we obtain exponential decay in the two point function which corresponds to exponential decay of clusters in the spatial dimensions. The decay in the continuous dimension, however, has a peculiar form which we state in the following condition. In Theorems \ref{rotationtheorem} and \ref{lowRecurrence} we state conditions sufficient to carry out the multiscale analysis and thus obtain the following type of decay.
 
 \begin{condition}[$\bfC_{\nu,m}$-localized]  
   Given $\nu,m > 0$, we say the ground state determined by a sampling function and coupling value  pair $(h,\lambda)$ is $\bfC_{\nu,m}$-localized if there is a full measure set $\cU_{\nu,m}\subset \bfTh $ so that the following holds. For all $\theta\in \cU_{\nu,m} $ and transverse field $\delta(x) = h(\bfT^x \theta)$, for all $x \in \bbZ^d$ there is some $C_x$ so that   for any $y\in\bbZ^d$ and $t,s\in\bbR$ so that $|x-y| + [\ln(1+|t-s|)]^{1/\nu}  >C_{x}  $ we have
\be \label{cnum}    \bfQ_{\delta,\lambda}^{(2)}  
             \left((x,t)\lra   (y,s)\right) < 
        \exp\left\{- m \left(|x-y| 
              + [\ln(1+|t-s|)]^{1/\nu}\right)\right\}.
  \ee 
 \end{condition} 
  We show below that sampling functions $h$ satisfying a transversality condition of the form (\ref{psilimit}) for $\chi_1 < 1/d$. On the other hand if $h$ also satisfies a condition of the form (\ref{psilimit2}) for some $\chi_2$ so that $\chi_1 >\chi_2 > 0$ the form of the decay (\ref{cnum}) is indeed the best possible, as we discuss in Section \ref{ooc}.

  We apply our analysis from the abstract setting to the one frequency quasiperiodic field in one dimension. A quasiperiodic environment process is defined by a rotation     $\bfT\theta = T_\omega\theta = \theta+\omega$ for $\omega \in [0,1]\setminus \bbQ$ acting on $\theta \in \bfTh  = \bbT$. To discuss our results we will introduce numerical properties of irrational numbers. We can write the irrational rotation $\omega$ in the continued fraction exapansion
  \be \label{cfracrep}  \omega =
     \frac{1}{a_1 +\frac{1}{a_2 +\cdots      }}
  \ee
   where $a_i$ are positive integers. We encode this expansion  as  $\omega = [ a_1,a_2,\cdots]$. Truncating to the $n^{th}$ term, we get the rational approximant $p_n/q_n = [a_1,\cdots,a_n]$. An irrational $\omega$ is said to be of finite type if $a_i$ are uniformly bounded, in particular, the Fibonacci number $\omega = (\sqrt 5 - 1)/2 $ is of finite type with $a_i \equiv 1$. On the other hand $\omega$ is $\gamma$-Diophantine for $\gamma>0$ if there is $C_\omega<\infty$ so that for all $n$ we have $q_{n+1} < C_\omega q_n^{1+\gamma}$. For any $\gamma > 0$, almost all real numbers are $\gamma$-Diophantine. By contrast, the set of finite type  $\omega$  compose a measure 0 set and are $0$-Diophantine. For finite type frequencies, we can specify a sharp condition on the transversality behavior at the zero set separating $F_{weak}$ from $F_{strong}$.

    \begin{theorem}\label{rotationtheorem}
      Let $\spc = \mathbb{T} = \bbR / \bbZ$ be the one dimensional torus.
       Let the transverse field $\delta$ be defined by the quasiperiodic sampling process $\bfT$, which is defined by  $\bfT^x\theta = \theta + x \omega $ for given $\omega \in \bbR\backslash\bbQ$.
   \begin{itemize}
       \item[(1.)]
       If $h \in \cC^{+}_{fin}(\bbT)$ is such that
       there is some point $\theta_0\in\mathbb{T}$ so that there exists $a > 1 $ so that
         \[   \liminf_{\epsilon \to 0} \inf_{\theta: r(\theta, \theta_0) < \epsilon }  \frac{ \log|\log h(\theta) | }{  |\log \epsilon  |}  >  a, \]
        then $h\in F_{strong}$ for  any irrational frequency $ \omega$.
        \item[(2.)]
        On the other hand, if $\omega\in \bbR\setminus \bbQ 
        $ is $\gamma$-Diophantine and there is an $a< \frac{1}{ 1 + \gamma}$ so that the sampling function $h$ obeys
         \[   \limsup_{\epsilon \to 0} \sup_{\theta: r(\theta, h^{-1}(0)) > \epsilon } \frac{ \log|\log h(\theta) | }{  |\log \epsilon  |}  <  a \]
        then $h\in F_{weak}$. Moreover,  for any $ \frac{\alpha (1 + \gamma) +1}{2}  <  \nu < 1$ and $ m> 0$, there is $\lambda_{m,\nu}$ so that for $\lambda < \lambda_{m,\nu}$, the pair  $h,\lambda$ satisfies $\bfC_{\nu, m}$-localization.
   \end{itemize}
    \end{theorem}
   Taking $(2.)$ with $(1.)$ demonstrates a critical disorder at $a = 1$ for rotations of finite type. The proof  of part $(2.)$ is also contained in \cite{jito}.

   These results hold as a corollary of the analysis in Section \ref{analysis} and the following general theorems establishing conditions on localization and long range order.

   The multiscale method relies on a transversality condition on $h$ at $h=0$. To evaluate the behavior at the zero set in the abstract setting, we define the following functions to compare the behavior of the sampling function $h$ with the recurrence defined by an environment process $(\bfTh,\bfT)$. Given $h$, let
   \begin{equation}\label{gdef}
      \phi_h(\epsilon) = \inf\{  h(\theta) : \forall \theta \in \bfTh  \tn{ so that }   r(\theta, h^{-1}(0))  >  \epsilon  \}
   \end{equation}  
   To quantify the recurrence  rate we introduce
   \[ K_1(\epsilon) = \inf_{ \theta\in \bfTh } \max\{K :r(\theta,\bfT^x\theta)> \epsilon \tn{ holds } \forall x \tn{ so that } 0< |x| \leq K   \} \]
 By aperiodicity $r(\theta,\bfT^x \theta) > 0$ for $x\neq 0$, so, by ergodicity and compactness of $\bfTh$, $K_1 \to \infty $ as $ \epsilon \to 0$. 
    The following theorem states a sufficient condition for a sampling function $h$ to admit a localized phase, 
   \begin{theorem} \label{lowRecurrence}
    Let $(\bfTh, \bfT)$ be an environment process.
     Suppose, for some $\chi  < 1/d$, a sampling function $h$  satisfies 
     \be\label{psilimit}    \limsup_{ \epsilon \to 0 }   \frac{\log |\log \phi_h(\epsilon)| }{\log K_1(\epsilon) }    <  \chi .     \ee
      Then for any $\nu$ satisfying $   \frac{1 + \chi }{1 + 1/d } < \nu <1 $ and  any $\infty > m > 0$ there is $\lambda_m$ so that for $ 0 < \lambda < \lambda_m$,     the pair $(h,\lambda)$ are $\bfC_{\nu,m}$-localized. 
   \end{theorem}

   From this condition, it is possible to show the set of  $C_{\nu,m}$-localized $(h,\lambda)$ pairs is nonempty.
      \begin{theorem}\label{fasterthanpolynomial}  
         Let $(\bfTh, \bfT)$ be an environment process. Let $A \subset \bfTh$ be a finite set. Let $m$ and $\nu$ be  parameters such that $ 0 < m  < \infty $ and $\frac{1}{1 + 1/d} < \nu < 1$. Then there exists a function $g_A\in \cC^{+}_{fin,1}(\bfTh)$ so that $g^{-1}_A(0) = A$ and for sufficiently small $\lambda$, the pair $(g_A,\lambda)$ is $\bfC_{\nu,m}$-localized. 
      \end{theorem}

   A complementary result establishes a condition for long range order.  Define, a second parameter for transversality at $h^{-1}(0)$:
     \begin{equation}\label{gdef2}
      \psi_h(\epsilon) =\min_{\theta_0 \in h^{-1}(0)} \sup\{ h(\theta) : r(\theta, \theta_0) < \epsilon  \}.
   \end{equation}  
   Define an alternative recurrence quantity for the minimal subset of the orbit required to cover $\bfTh $ with $\epsilon$ neighborhoods:
   \[  K_2(\epsilon)  = \sup_{\theta_0,\theta_1\in \bfTh}
                        \min\{ K : \exists x \textnormal{ so that } 0 \leq |x| \leq K \textnormal{ and } r(\theta_1,\bfT^x \theta_0) < \epsilon \}.   \]
  As $\bfT$ is ergodic, $K_2$ is finite for all $\epsilon > 0$.
  \begin{theorem} \label{hirec}
      Let $(\bfTh, \bfT)$ be an environment process.
     Suppose there is  some $\chi   >  1$, so that
     \be\label{psilimit2} \limsup_{ \epsilon \to 0 }   \frac{ \log|\log \psi_h(\epsilon)| }{\log K_2(\epsilon) }   > \chi ,     \ee
       then  $h \in F_{strong}$.
   \end{theorem} 

 Now from Theorem \ref{hirec} we can show the set $F_{strong}$ is nonempty.

\begin{theorem}\label{strongfun}
  Let $(\bfTh,\bfT)$ be an environment process and let $\theta_0 \in \bfTh$. There exists a function $ f_{\theta_0} \in \cC_{fin,1}^+(\bfTh)\cap F_{strong}$ such that $f_{\theta_0}^{-1}(0) = \theta_0$. Moreover, given $\epsilon > 0$, $f_{\theta_0}$ may be chosen such that $f_{\theta_0}(\theta) = 1$ for all $\theta$ such that $r(\theta_0,\theta) > \epsilon$.
\end{theorem}

   \subsubsection{Phase transition in the random case} \label{randomcon}
   Conditions for localization in the random case were first investigated by  Campanino and Klein \cite{cam91}.
   In that paper, a multiscale argument demonstrated localization for sufficiently nice  i.i.d. parameters $\lambda_{(x,y)}$ and $\delta_x$.  The   argument was optimised by Klein in \cite{klein94}. There, it is shown that, if for some $\alpha > \alpha_d := 2d^2(1 + \sqrt{1 + \frac{1}{d}   } + \frac{1}{2d}  )$,  the moment conditions 
   \begin{equation}\label{kleincondition1}
     \left\langle \ln\left(1 + \frac{1}{\delta}\right)^\alpha\right\rangle
        <\infty;\hspace{.5in}
         \langle\ln(1 +\lambda)^\alpha\rangle <\infty
   \end{equation}
   are satisfied, then the ground state is almost surely short range ordered
    provided a   low density assumption 
   \begin{equation}\label{kleincondition2}
      \left\langle\left(\ln\left(1 + \frac{\lambda}{\delta}\right)
         \right)^\alpha\right\rangle<\epsilon,
   \end{equation} 
   holds  for  sufficiently small $\epsilon > 0$. 
   Moreover, in this regime (\ref{order}) almost surely has exponential decay. For constant $\lambda$, (\ref{kleincondition1}) reduces to a condition on the distribution of $\delta$ near zero. Thus, for sufficiently small $\lambda$, the ground state is in the  short range phase, so that (\ref{kleincondition1}) is a condition for  weak disorder.

   On the other hand, Aizenman, Klein and Newman
    \cite{AizKleinNewman} investigated conditions 
   for $\delta$ so that the ground state is in the long 
   range phase for all $\lambda >0$.
   As these distributions eliminate the phase transition
   they are known as strong disorder.
   For $d = 1$ the strong disorder condition is
   \begin{equation} \label{AKNcondition1}
         \lim_{u\to\infty} \frac{u}{|\ln(u)|}
       {\bf P}\left(\left\{ \ln\left(1 + \frac{1}{\delta}\right) 
          > u \right\} \right)  = \infty;\hspace{.5in}
		 {\bf E_P}\left( \delta + \frac{1}{\lambda}  \right)  < \infty .
    \end{equation}
For $d\geq 2$  for $\delta$ satisfying
   \begin{equation}\label{AKNcondition2}  
    \lim_{u\to\infty} u^{d}{\bf P}
        \left(\left\{ \ln\left( 1 + \frac{1}{\delta}\right) 
          > u \right\} \right)  = \infty
   \end{equation} 
   the ground state is long range ordered for all 
   constant $\lambda > 0$.
   
   As indicated by Theorems \ref{introtheorem} and \ref{introtheorem2} the disorder conditions are mirrored in  aperiodic systems. There is a short range phase for small $\lambda$ for generic sampling functions  and a persistent long range phase for all $\lambda > 0$ only for pathological behavior of the sampling functions near zeros.
 
    Let us compare conditions (\ref{kleincondition1}) to (\ref{AKNcondition2}) in the random case  to conditions in Theorem (\ref{rotationtheorem}) to (\ref{hirec}) in the quasiperiodic case.  In the quasiperiodic system,  discussed in Theorem \ref{rotationtheorem}, the ergodic measure $\mu$ is just the Lebesgue measure. Roughly, conclusion (1.) states a sufficient condition for $h \in F_{strong}$ corresponds to existence of an $a > 1$ so that, for some constant $c>0$ and for sufficiently small $\epsilon$,  $ \epsilon^{-1}  \mu ( \{\theta: | \log  h^{-1}(\theta)  | > \epsilon^{-a} \} ) > c $. It follows that $ \epsilon^{-1}  \mu ( \{\theta: | \log  h^{-1}(\theta)  | > \epsilon^{-1} \} ) \to \infty  $ as $\epsilon \to 0$  which is similar to the iid case.
    On the other hand, conclusion (2.) implies that $h$ in which admits a localized phase 
    corresponds to $\mu ( \{\theta:  \log [ 1 + h^{-1}]  > m \}  ) < m^{-b} $ for some $b> 1$.
    Thus if $ b> \alpha > 1$ then $\mu ( \{\theta:   \log^\alpha [ 1 + h^{-1}]  > m \}  ) < m^{- b/\alpha} $ 
    then $\bbE[ \log^\alpha ( 1 + h^{-1})     ] < \infty $ which corresponds to (\ref{kleincondition1}). 
    \label{randomdyncom}

  \subsubsection{Organization of the paper}
    The rest of the paper is organized as follows. In Section \ref{models} we  introduce the  random cluster model on $\bbZ^d\times \bbR$ and relate it to the quantum Ising model. In Section \ref{proofs} we introduce and study the properties of $\wh M$ a connectivity parameter in the percolation model. In the context of $\wh M$, we prove the 0-1 claim in Section \ref{reg2} and we construct the $G_\delta$ set in Section \ref{Borel}. In Section \ref{Secmainproof}, we establish an equivalence of  Theorems \ref{introtheorem} and \ref{introtheorem2}. We complete the section by constructing the necessary functions to demonstrate Theorems  \ref{fasterthanpolynomial} and \ref{strongfun}.

     We introduce the multiscale analysis in Section \ref{msa}, which states $\bfC_{\nu,m}$-localization follows from a well behaved environment. In  Section \ref{low}, we demonstrate  Theorem \ref{lowRecurrence} by showing the condition (\ref{psilimit}) is sufficient to demonstrate the enviroment is well behaved. In Section \ref{ooc} we prove Theorem \ref{hirec} by constructing an auxiliary bond-site percolation model coupled to the original percolation model, we show an infinite cluster in the bond-site model implies an infinite cluster in the original model. Finally, we conclude the paper applying the results to the phase transition in the quasiperiodic case and prove  Theorem \ref{rotationtheorem} in Section  \ref{rotations}.

\section{FK representation}\label{models}

     Quantum spin models may be related by a Fortuin Kasteleyn representation to a percolation processes called the random cluster model. The percolation model takes place on $\bbZ^d \times \bbR$. Measurements of observables in the Ising model, such as the correlation function $ \langle \sigma_x^{(3)} \sigma_y^{(3)} \rangle_{\delta,\lambda}  $ are equal to communication probabilities in the random cluster model.  

      The random cluster model can be defined in terms of a product measure percolation model, called continuous-time percolation \cite{bez91}, which we will introduce first. Despite the name, the `time' dimension in the model is non-oriented. The oriented version of continuous-time percolation is the well known contact process. Essentially, allowing communication in the negative time direction recovers   continuous-time percolation. Moreover, there are stochastic dominations between the random cluster model and continuous-time percolation, which will be useful for demonstrating localization and percolation in various regimes.

   \subsection{Continuous-time percolation}
   We begin with the graph $\lattice = \left(\zd,\bbE^d\right)$, where $\bbE^d$ is the set of nearest neighbor pairs $\{x,y\}$ in $\zd$ so that  $\|x-y\|_1 = 1$. The parameters $\delta: \bbZ^d \to (0,1)$ and $\lambda > 0$ define an environment. For every $x\in \zd$ there is a Poisson process of deaths on $\{x\}\times\real$ at rate $\delta(x)$, and for every edge $\{x,y\}\in\bbE^d$, there is a Poisson process of bonds at a rate $\lambda$ on $\{ x,y\}\times\real$.
   The measure of the Poisson process of deaths on $\{x\}\times\real$
   will be denoted as $\overline\bfQ_{\delta;\lambda}^{x}$,
   similarly the Poisson process of bonds for any 
   $u\in\bbE^d $ will be labeled $\overline\bfQ_{\delta;\lambda}^{u}$.
   The space of realizations 
   for the Poisson measures $\overline\bfQ_{\delta;\lambda}^{*}$
   for each $*\in \bbZ^d$, respectively $*\in\bbE^d$, is the set of all
   locally finite sets of points in $\{x\}\times\real$
   denoted $\cons_x$,
   respectively locally  finite sets in $ u \times\real$
   denoted $\cons_u$.
   Any locally finite set $\omega$ in 
   $\left(\zd\cup\bbE^d\right) \times \real$
   is called a configuration,
   and the space of all configurations
   is denoted $\cons$,
   it is the product of all sets $\cons_*$,
    that is  $ \cons = 
        \left(\times_{x\in\zd} \cons_x\right)\times
        \left(\times_{u\in\bbE^d} \cons_u\right).$
   The percolation measure on $\cons$ 
   is the product measure of these Poisson processes,
   \begin{equation}\nonumber
    \csm{\delta}{} =    
        \left(\prod_{x\in\zd} \overline\bfQ_{\delta;\lambda}^{x}\right)
        \left(\prod_{u\in\bbE^d} \overline\bfQ_{\delta;\lambda}^{u}\right).
   \end{equation}
  When $\delta$ is defined by a sampling function $\delta(x) = h(\bfT^x\theta)$ we use the notation 
  $\bfQ_{h,\theta;\lambda} =\bfQ_{\delta;\lambda}  $
  For given $\omega$ let us denote the set of deaths 
  by $\bfD_\omega$ and the set of bonds by $\bfB_\omega$. 
  Formally, elements of $\bfD_\omega$ are singletons of $ \bbZ^d \times \bbR $ and elements $\{(x_1,t_1),(x_2,t_1)\} \in \bfB_\omega  $ are subsets of  $ \bbZ^d \times \bbR $  order 2 so that $\{x_1,x_2\} \in \bbE^d$ and $t_1 = t_2$.

    The proper topology of $\cons$ is the Skorohod topology, roughly speaking, a configuration $\omega$ is close to $\omega'$ if bonds and cuts in the configurations are close on bounded sets. Let $\cF$ be the the $\sigma $-algebra generated by the Skorohod topology. For a discussion of the Skorohod topology in the present context, see \cite{Grimmett10}. See  \cite{ethier} for a thorough background of the Skorohod topology.

   Any configuration $\omega$ induces a partition on 
    $(\bbZ^d \times \bbR) \setminus \bfD_\omega$. If $ W = x\times (t_1,t_2) \subset \bbZ^d \times\bbR $ is an interval so that $W\cap \bfD_\omega = \emptyset$, then the points of $W$ belong to the same component of the partition. For any $W \in \bfB_\omega $ the two points of the $W$ belong to the same component of the partition.
  Two points $X,Y\in\lxt$  communicate in $\omega$
   if there is a path  
   \be \label{compath}   W =  (W_1,W_2,\ldots,W_m) \ee
   so that, $X\in W_1,$ $Y\in W_m$, and  each $W_i $ is either an interval not intersecting a death or a bond. We write $C_\omega(X) $ for the component of the partition containing $X$,  and $X \lra Y$ if $Y \in C_\omega (X)$.

    Note that communication is permitted in both the negative and positive `time' direction. If communication is permitted only in the positive time direction, the model becomes the well known contact process.
   
   Here we will introduce some notation and terminology for bounded sets. Let $\Lambda_L(x) = \{y\in \bbZ^d : \|y-x\|\leq L\}$. Let
    Cylinders are sets of the form $B= W \times I$ for $W \subset\subset \bbZ^d$ and $I = [a,b]$. In particular we write $B(n) = \Lambda_n(0)\times [-n,n]$. The boundary of $W\subset \bbZ^d$ is defined as
       \[ \partial W = \{y\in W|\exists y' \in W^c
          \textrm{ so that } \langle y,y'\rangle\in \edge\}. \]
    Given a cylinder $B$,  the horizontal boundary of is $ \partial_H B = W\times\{a,b\}  $
     and the vertical boundary is $  \partial_V B = \partial W \times I$.
     The boundary of a cylinder $B$ is then $ \partial B = \partial_H B \cup \partial_V B  $. 
     
    Given a cylinder $B$ we say $X$ communicates with $Y$ in $B$ if there is a path (\ref{compath}) which is entirely contained in $B$. We write $C_{B,\omega}(X)$ for the component of the partition  of $B\setminus \bfD_\omega$ containing  $ X$.  A boundary condition of a cylinder $B$ is a collection of non intersecting Borel measurable subsets of the boundary $\partial B$. Points within the same subset of the boundary condition are declared to belong to the same component in the percolation sense described above. Formally, given a partition $\bfb$ of  $B$ we say $X$ communicates  with $Y$ with respect to boundary condition $\bfb$, if $C_{B,\omega}(X) = C_{B,\omega}(Y)$ or if there is an element $W\in \bfb $ so that $ C_{B,\omega }(X) \cap W \neq \emptyset $ and  $ C_{B,\omega }(Y) \cap W \neq \emptyset $. Finally, we write $ C_{B,\omega}^{\bfb}(X) $ for the set of points $Y$ communicating with $X$ in $B$ with respect to the boundary condition $\bfb$. 
    
    The most important boundary conditions are the `free', `periodic', and `wired' boundary conditions. In the wired boundary $\bfw = \{ \partial B\}$, all points of the boundary are declared to communicate. On the other hand,  the free boundary is  empty: $\bff = \emptyset $. Thus, for any $\omega$, and $X \in B$ the components of $X$ under the free boundary condition are exactly those which communicate within the cylinder: $ C_{B,\omega}^{\bff}(X) = C_{B,\omega} (X)$.   
 For a  cylinder $B = \Lambda \times [-\beta , \beta]$, the periodic boundary condition is,
     \[  \bfp =
       \{ \{(x,\beta),(x,-\beta)\}:x\in \Lambda \}   \]
  which may be visualized as constructing the percolation environment on $\Lambda \times 2\beta \bbT$.

  \subsection{Random cluster measures}    

      The  random cluster measure is first introduced on bounded cylinders  $B \subset\subset \bbZ^d\times \bbR$. For $\omega\in \cons$, let  $\omega_B =  B\cap( \bfD_\omega \cup \bfB_\omega)$ be the restriction of $\omega$ to $B$, and let $\cons_B$ be the set of all such $\omega_B $. For $B \subset \bbZ^d \times \bbR$ we write $\cF_B$ for the $\sigma$-algebra of events generated by the Skorohod metric which depend only on $\bfD_\omega \cap B$ and $\bfB_\omega \cap B$. We will now construct the random cluster measure on $(\cons_B,\cF_B)$.

    We define  $k_B^{\bfb}:\cons\to\zee^+$ to be  the function counting the number of clusters in $B = \Lambda \times [-T,T]$ with respect to boundary condition $\bfb$.  $k_B^{\bfb}$ is almost surely finite for compact $B$ and therefore well defined. Indeed, the number of clusters is bounded by $|\Lambda|  +  |\bfD \cap B| $, so letting $K = \max_{x\in \Lambda} \delta(x)$, we see $ k_B^\bfb - |\Lambda| $ is bounded by a Poisson random variable with rate $2 TK |\Lambda|$.
    Moreover,  for $q > 0 $, 
   $ q^{k_B^{\bf}(\cdot)}\in L_1(\cons|B, \csm{\delta}{ }|_B)$,  
   \begin{equation}\label{L1}
      \bfQ_{\delta;\lambda}(q^{k_B^{\bfb}})\leq 
      \bfQ_{\delta;\lambda}\left(q^{ |\bfD \cap B|  + |\Lambda|}\right) \leq q^{|\Lambda|}\exp\{(q - 1)\times  2 T K  |\Lambda| \}.
   \end{equation} 
   Thus, given $B$ with boundary condition ${\bfb}$, we can define the continuous time random cluster measure on $\cF_B$  by  
   \be  \label{rcm2} \bfQ_{\delta;\lambda}^{(q)}|_{B}^{\bfb}(A) :=
         \frac{\int_{\cons_B} {\bf 1}_A(\con) q^{k_B^\bfb(\con)} d\bfQ_{\delta;\lambda}(\con) }
              {\int_{\cons_B} q^{k_B^\bfb(\con)} d\bfQ_{\delta;\lambda}(\con) }   .\ee
  for $A \in \cF_B$. Note that if we set $q=1$ we recover the independent percolation model.

  \subsection{Classical Ising model in $d + 1$ dimensions}\label{clasd}
   
     The random cluster measure will assist us in calculating quantities such as the two point function (\ref{twopoint}). 
   
     We fix a set $\Lambda \subset \subset \bbZ^d$ and finite inverse temperature $0 < \beta < \infty$. As above, associate the periodic boundary condition $\bfp$ to the cylinder $B = \Lambda \times [-\beta , \beta]$. For any $\omega \in \cons$ we say a map $\sigma: B \to \{1,-1\}$ is conditioned to $\omega$ and $\bfp$  if $\sigma $ is constant on all clusters $C_{\omega,B}^{\bfp} (X)$ for $X \in B$. We write $\Gamma_{\omega,B}^{\bfp}$ for the set of such maps. Note that functions $\sigma$ here are configurations of the  classical Ising model.
   Recall the single site Pauli matrices are  $\sigma^{(1)} = \begin{pmatrix} 0&1\\1&0 \end{pmatrix}$,
     $\sigma^{(2)} =  \begin{pmatrix} 0&-i\\i&0 \end{pmatrix} $,
      and  $\sigma^{(3)} =   \begin{pmatrix} 1& 0 \\0&-1 \end{pmatrix} $.
   We wish to consider expectations of monomials
  \be \label{monomial}   \left\langle \prod_{i}  \sigma_{x_i}^{(a_i)}  \right\rangle^{(\Lambda,\beta)}_{\delta,\lambda} \ee
    where $x_i \in \Lambda$ and $a_i \in \{1,2,3\}$.
   It is convenient to introduce the auxilliary single site operator 
   \[   \sigma^{(0)} = \begin{pmatrix} 1&1\\1&1 \end{pmatrix}. \] 
   Indeed, $\sigma^{(1)} = \sigma^{(0)} - I$ and $\sigma^{(2)} = i(\sigma^{(0)} - I) \sigma^{(3)}$ so we can reduce (\ref{monomial}) to polynomials with $a_i \in \{0,3\}$. To avoid an ambiguity in the FK construction we replace each factor of $\sigma^{(3)}_x$ with a factor  of $ e^{\epsilon H_\Lambda}\sigma^{(3)}_x e^{- \epsilon H_\Lambda}$ and take the $\epsilon \searrow 0$ limit.
   A factor of $e^{\epsilon H_\Lambda}\sigma^{(3)}_x e^{- \epsilon H_\Lambda}$ yields a factor of the sign of $\sigma(x,\epsilon)$ in the FK expansion.  Placing a factor of $\sigma^{(0)}_x$ allows a change of sign of $\sigma(x,t)$ at $ t = 0$. Thus, given $\omega$ we we define $\omega_{y_1,y_2,..y_n}$ the configuration with bonds $\bfB_\omega $ and deaths $\bfD_\omega \cup \{(y_1,0),..., (y_n,0)\}$. For monomial (\ref{monomial}) let $\bfx $ be the sublist of $(x_i)$ so that $a_i = 0$ and $\bfy$ be the sublist of $(x_i) $ so that $a_i = 3$. The FK representation states \cite{AizKleinNewman},  
     \be \label{mono2fun} 
      \left\langle \prod_{i} 
       \sigma_{x_i}^{(a_i)}  \right\rangle^{(\Lambda,\beta)}_{\delta,\lambda} 
       = \frac{1   }
     {\int_\cons 2^{k_B^\bfp(\con)} d\csm{\delta}{ }(\con) }
     \lim_{\epsilon \searrow 0}    
     \int_\cons  \sum_{\omega \in \Gamma_{\omega_\bfx}^\bfp} 
           \prod_{x\in \bfy} \sigma(x,\epsilon) d\csm{\delta}{ }(\con),
            \ee 
     Thus, expectations of  polynomials in $\bbC[\sigma^{(a)}_x] $ can be computed from the measure (\ref{rcm2}) (\ref{monomial}) can be calculated by Skorohod measurable functions on $\cons$. In particular we observe that   
         \be  \label{qim rcm rel}
             \langle \sigma_x^{(3)} \sigma_y^{(3)} \rangle^{(\beta , \Lambda)}_{\delta,\lambda} =    \csm{\delta}{(2)}|^\bfp_{\Lambda \times [-\beta,\beta]}\{ (x,0) \llra (y,0)  \}.
    \ee   
     In the next section, we show the existence of the limit of $\csm{\delta}{(2)}|^\bfp_{\Lambda \times [-\beta,\beta]}$ as $\Lambda \to \bbZ^d$ and $\beta \to \infty$.
   \subsection{Positive events}\label{posy}

   The set of configurations, $\cons$, enjoy a partial ordering property. For $\con,\con'\in\cons$, if $\bfD_{\omega'} \subset \bfD_{\omega}  $ and $ \bfB_{\omega } \subset \bfB_{\omega'} $, then we write the ordering  $\con \leq \con'$. A set $U\subset\cons$ is said to be positive if $\con\in U$ and $\con \leq \con'$ together imply $\con'\in U$. Notice communication events are positive, as $C_\omega (X) = C_\omega( Y)$ and $\con \leq \con'$ imply $ C_{\omega'}(X) =  C_{\omega'}(Y)$.

   Measures ${\bf Q}$ 
   on $\cons$ enjoy a partial ordering property as well called
   stochastic ordering.
   If, for all positive sets $U$,
   $ {\bf Q}(U) \leq {\bf Q'}(U)$ we say ${\bf Q'}$ dominates ${\bf Q}$  and we write $ {\bf Q} \leq {\bf Q'}$. On the other hand, an event $U $ is negative if $U^c$ is positive, we say $U$ is monotonic if it is either positive or negative.

    The random cluster models may be  bounded above and below by independent percolation models, as we see in the following proposition. 
      \begin{proposition}\label{RCMordering} Let $B \subset \bbZ^d \times \bbR$ and let $\bfb$ be any boundary condition.
         Let  $ q \geq q'\geq 1$, $\lambda,\lambda' \in [0,\infty)$ and $\delta,\delta':\zd\to[0,\infty )$. If $ \lambda' \geq \lambda$ and $ \delta'\leq\delta$ we have the ordering,
         \begin{equation}\label{order1}
          {\bf Q}^{(q)}_{\delta;\lambda}|_{B}^\bfb \leq {\bf Q}^{(q')}_{\delta';\lambda'}|_B^\bfb 
         \end{equation}
         On the other hand, if $ \lambda'\leq \lambda q'/q$ and $\delta'\geq\delta q/q'$
         \begin{equation}\label{order2}
            {\bf Q}^{(q')}_{\delta';\lambda'}|_B^\bfb \leq {\bf Q}^{(q)}_{\delta;\lambda}|_B^\bfb 
         \end{equation}
      \end{proposition}

   This ordering is similar to the ordering for the measure of the  discrete random cluster model \cite{G2010},
   a proof of  Proposition \ref{RCMordering} in the continuous context 
   can be found  in \cite{bjornberg} as Theorem 2.2.12.
   For our purposes, communication events $A = \{\omega : C_\omega(X) = C_\omega(Y)\}$ are the relevant positive events. Notice that (\ref{order1}) and (\ref{order2}) allow us to bound, above and below, percolation events in the random cluster model with percolation events in axillary independent percolation models. 

    We will also require the FKG inequality bounding probabilities of intersections of positive events. We call an event $U \subset \cons$ an event of continuity, if $\partial U$ is a set of measure 0 with respect to the independent percolation measure $\bfQ_{\lambda;\delta}$. Notice the random cluster measures $\bfQ_{\lambda;\delta}^{(q)}|_B^{\bfb}$ are continuous with respect to the independent percolation measures $\bfQ_{\lambda;\delta}^{(1)}$.      
    
    \begin{proposition} \label{pfkg}Let $B = \Lambda \times I$ be a bounded cylinder and let $\bfb$ be any boundary condition. Let $U,V$ be monotonic events of continuity. For any non-negative $\delta: \Lambda \to \bbR^+ $, non-negative $\lambda \geq 0$, and  any $q \geq 1$ we have
   \begin{equation}\label{fkg}
   \csm{\delta}{(q)} |_B^{\bfb} \left(U\bigcap V\right) 
         \geq \csm{\delta}{(q)} |_B^{\bfb}  (U)
         \csm{\delta}{(q)} |_B^{\bfb}  (V).
   \end{equation} 
   \end{proposition}
   This result is similar to the inequality for the discrete random cluster model \cite{G2010}, it appears in the continuous context as Theorem 3.1 of \cite{AizKleinNewman}.
   \begin{remark}
    For any $B \subset \bbZ^d \times \bbR$ and boundary condition $\bfb$, if $U = \{C_B^{\bfb}(X) = C_B^{\bfb}(Y)\}$ is a communication event then $\partial U$ has measure zero, which we prove in Proposition \ref{zerobdry}. Thus communication events are events of continuity.  
    \end{remark}

   Finally, let us state the existence of the infinite volume random cluster measure. For the selected sets $B(k) =\{(x,t): \|x\|_\infty \leq k; \; |t| \leq k\}$,  let us write  $\cF_k \equiv \cF_{B(k)}$ and write the  tail $\sigma$-algebra  as $\cF_\infty = \cap_n \sigma( \cup_{k\geq n} \cF_{k} )$. Let $\bfQ_k$ be a measure defined on $(\cons_{B(k)},\cF_k)$. Let $\bfQ$ be a measure on $(\cons,\cF)$. We say $\bfQ_k \to \bfQ$ weakly, if for any event of continuity $U\in \cF_k$, for some $k$,  $ \lim_{n \to \infty}\bfQ_n(U) \to \bfQ(U)$.    
   
   In \cite{bjornberg} the following appears as Theorem 2.3.2 for $\bfb = \bfw,\bff$, but the proof is similar for $\bfb = \bfp$.
      \begin{proposition}\label{weaklim} 
         Let $\lambda \in \bbR$ be non-negative and let $\delta : \bbZ^d \to \bbR$ be a non-negative bounded function. For $\bfb = \bff,\bfp,\bfw$, there is a limiting measure $\csm{\delta}{(q)|\bfb} $, so that  $\csm{\delta}{(q)}|^{\bfb}_{B(n)}\to\csm{\delta}{(q)|\bfb}$ weakly.  
      \end{proposition}
    In fact, the limiting measure  is independent of the sequence of sets $B(n)$.   
      
  As discussed in  Section \ref{clasd}, we are interested in $q = 2$. For bounded non-negative  $\delta$, the sequence of measures $\csm{\delta}{(2)}|^\bfb_{B(n)}$ converges weekly to a measure $\csm{\delta}{(2)|\bfb}$. Combining (\ref{mono2fun}) and Proposition \ref{weaklim} we can now find the infinite volume expectation of polynomials $\bbC[\{\sigma_x^{a}\}_{x\in \bbZ^d, a=1,2,3}]$ which is sufficient to determine the expectation of all local operator and therefore define the ground state.  We can now formally define our two point function as the limit
   \begin{align} \label{fl2}  \langle \sigma_x^{(3)} \sigma_y^{(3)} \rangle_{\delta;\lambda} &  := \lim_{\beta \to \infty, \Lambda \nearrow \bbZ^d}   \langle \sigma_x^{(3)} \sigma_y^{(3)} \rangle^{(\beta , \Lambda)}_{\delta;\lambda} \\
 \nonumber
  & 
  = \nonumber        
   \lim_{\beta\to\infty,\Lambda\nearrow \bbZ^d} \bfQ_{\delta;\lambda}^{(2)}|^\bfp_{\Lambda \times [-\beta,\beta]}\left( C_{\Lambda\times [-\beta,\beta]}^{\bfp}(x,0) = C_{\Lambda \times [-\beta,\beta]}^\bfp(y,0) \right).    \end{align}

\section{Regularity of random cluster measures}\label{proofs}
   
    The main goal of this section is to show the Borel regularity of $F_{weak}$. Essentially, the strategy is to show regularity of the percolation measures with respect to $\delta$ and $\lambda$ in finite subsets and extend these properties to the infinite lattice.

   We write probability of communication from the origin to the boundary of $B(n)$ as
   \be \wh M^{(q)|n}_{h,\theta;\lambda} 
   :=  \bfQ^{(q)}_{h,\theta;\lambda} 
   ( \{0,0\}  \lra \partial B(n )       ), \label{mboxn} \ee
   which is clearly decreasing, so the limit
   $\wh M_{h,\theta;\lambda}^{(q)} := \lim_{n\to \infty}
    \wh M^{(q)| n}_{ h,\theta;\lambda} $ exists.
     For fixed $(h,\theta, \lambda)$ we define the set of phases defining absence of long range order to be
  \[ \wh\cU_0  =  \wh\cU_0(h,\lambda) = \{ \theta \in \bfTh : \wh M^{(q)}_{h,\theta;\lambda}   = 0   \}. \] 
   Finally we define $\wh F_\cdot$ similar to $F_\cdot$ for $\cdot \in \{strong, weak\}$. Let $\wh F_{weak}$ be the set of sampling functions in $\cC_{fin}^+(\bfTh)$ such that, for sufficiently small $\lambda > 0$, $\mu(\wh \cU_0) = 1$. Let $\wh F_{strong}$ be the set of sampling functions so that, for all $\lambda > 0$, $\mu(\wh \cU_0) = 0$.
   The behavior of $\wh M^{(2)}_{h,\theta;\lambda}$ is similar to the behavior of $M_{h,\theta;\lambda}$ which we show in Proposition \ref{mmhat compare}. First we state the $G_\delta$ result in the random cluster setting.
     \begin{proposition}\label{gdeltatop}
            For any $q \geq 1$, $\wh F_{weak}$ is a $G_\delta$ in the $(\cC_{fin}^+ (\bfTh) , d_\tau)$ topology.
      \end{proposition} 
   
  \begin{remark} 
   It is possible to derive a similar result without requiring normalization of functions using $d_\tau$. However, in this case we would require the topology to include the coupling parameter. To do this we would define a metric $ d_{\infty}^\sharp( (h,\lambda), (h' , \lambda') ) = d_\infty (h,h') + |\lambda - \lambda'| $.
          Then, a similar proof would show that the set of $(h,\lambda) \in \cC_{fin}^+ \times \bbR^+$   so that   $\mu(\cA_1(h,\lambda))  = 1$ is a $G_\delta$ in  the $(\cC_{fin}^+ \times \bbR^+,   d_\infty^\sharp)$ topology.
  \end{remark}    
    
   Let us define the metric
   \be \label{dnatural}
    d_\infty^\natural( (h,\theta,\lambda), (h',\theta', \lambda') ) =  d_\infty(h, h')   +  |\lambda - \lambda'|+ r(\theta,\theta')\ee 
   which allows us to discuss a notion of continuity of environments.

   \begin{proposition}\label{mreg}
       For any $q \geq 1$, $ \wh M^{(q)}_{h,\theta;\lambda} $ is upper semicontinuous in $d_\infty^\natural$. Moreover, for any $q\geq 1 $, $\lambda>0$, $h  \in \cC^+_{fin} (\bfTh) $,  we have  $\mu(\wh\cU_0(h,\lambda)  ) = 0  $ or $1$.   
   \end{proposition}
 
    Proposition \ref{mreg}, is demonstrated at the end of Section \ref{reg2}.   

   \subsection{Continuity of measures}\label{regq}
      This section is devoted to showing  communication events  
        for fixed $ q\geq 1$ are continuous in the space of continuous sample functions,  $\cC^+_{fin}(\bfTh)$. 
       The culmination of the results of this section is
       contained in Proposition \ref{weakc} which follows from
       Lemma \ref{lxtconv}  and Proposition \ref{zerobdry}.

      \begin{proposition}\label{weakc}
        Let  $ A = \{ W_1 \llra W_2\}$ for bounded $W_1$ and $W_2$ be a communication event. 
        Given $(h,\theta,\lambda)$ and $\epsilon$, there is some $\eta > 0 $ so that 
        $d_\infty^\natural ((h,\lambda,\theta),(h',\lambda',\theta')) < \eta$ implies that 
        \[ \left|\bfQ_{h',\theta'; \lambda'}^{(q)}(A) -  \bfQ_{h,\theta; \lambda}^{(q)}(A) \right| < \epsilon \]  
      \end{proposition}

   \subsubsection{Regularity with respect to parameters}

       Consider bounded sets $\Lambda \subset \bbZ^d$. 
       Let us introduce a topology on the environments  $(\delta, \lambda)$ where
       $\delta: \Lambda \to (0,\infty) $ 
     $ \lambda > 0$, 
    \be \label{localmet} \|( \delta,\lambda) - (\delta' ,\lambda') \|_\ell =   |\lambda - \lambda'|  + \sup_{x\in \bbZ^d}  |\delta'(x) - \delta(x)| 2^{-|x|}  
    \ee
     the local convergence of environments.  
        We will show convergence in the measures of communication events under this metric.

      The following propositions and their proofs are similar in spirit to \cite{G2010}  (also see \cite{bjornberg} for development in the continuum case). We  need some modifications since we are considering continuity in environments. 
      
      For this section we will make a boundedness assumption on a sequence of environments. Let $K > 0$ and for every $k \in \bbN$ let $\delta_k$ be a function $\delta_k : \bbZ^d \to [0,K]$. Let $\lambda_k$ be a bounded sequence in $[0,K]$. Recall $B_n = \Lambda_n(0) \times [-n,n]$. 
      
      \begin{proposition}\label{tightfamily} For $\bfb = \bff,\bfp,\bfw$, the measures $\bfQ_{\delta_k;\lambda_k}^{(q)}|^{\bfb}_{B(n)}$ form a tight family.
      \end{proposition}
      \begin{proof} 
         Recall that a family of measures $\cM$ on $\bfOm$ is  
         a tight family if, for every $\epsilon > 0$, 
         there exists a set $R_\epsilon\subset\bfOm$, compact in the Skorohod topology, so that
         \[ \inf_{P\in\mathcal{M}} P\left(R_\epsilon\right) > 1-\epsilon .\]
         We have $\delta_k(x) < M$ for all $k$ and $x$, by   Proposition \ref{RCMordering} for all $k$, $ \bfQ_{\delta_k;\lambda_k}^{\bfb} \geq \bfQ^{(q)|\bfb}_{\delta_k;\lambda_k} \geq \bfQ^{\bfb}_{qK;\lambda_k / q} $.

         We introduce a function $\chi_{u}:\bbZ\to\real^+$, for all $u\in\bbZ^d\cup \bbE^d$, to be specified. Let $V_x'\subset\cons|_{\{x\}\times\real}$ be the event that, for all $j$, deaths in $\con\cap\left(\{x\}\times[j,j+1)\right)$ are separated by at least $\chi_x(j)$. Similarly, define  $V_{\langle x,y\rangle}'$ the set with bonds spaced out by $\chi_{x,y}(j)$. The closure $V_u$ of $V_u'$ is compact \cite{ethier}(Theorem 3.6.3). For $\chi_x(j)$ decreasing quickly enough $V_x$ has probability greater than $1 - \frac{1}{(1+ |x|)^{d+2}}\epsilon$; similarly let $\chi_{\langle x,y\rangle}(j)$ decrease quickly so that (let $|x|\geq|y|$) $V_{\langle x,y\rangle}$ has probability greater than $1-\frac{1}{(1+|x|)^{d+3}}\epsilon$.

         By the ordering of measures and the observation that  $V_{\chi;\bbZ^d} = \cap_{x\in\zd}V_x$ is an increasing set, we have for all $n$,
         \[  \bfQ_{\delta_k,\lambda_k}^{(q)}|^\bfb_{B(n)} (V_{\chi;\bbZ^d} )\geq \bfQ_{qK; \lambda_k}^\bfb (V_{\chi;\bbZ^d} ) \geq 1- O(\epsilon). \]
        On the other hand  $V_{\chi;\bbE^d} = \cap_{e\in\bbE^d}V_e$ is decreasing, so
        \[  \bfQ_{\delta_k;\lambda_k}^{(q)}|^\bfb_{B(n)}
              (\bfOm \backslash V_{\chi;\bbE^d})
              \leq \bfQ_{\delta_k;\lambda_k}^\bfb
              (\bfOm\backslash V_{\chi;\bbE^d}) \leq O( \epsilon). \]
         Finally the set $V_\chi = V_{\chi;\bbZ^d} \cap V_{\chi;\bbE^d}$   is compact and  $\bfQ^{(q)}_{\delta_k;\lambda}|_{B(n)}^\bfb(V_\chi) \geq 1- O(\epsilon)$ for all $k$ and $n$.
      \end{proof}
      We will show $\bfQ_{\delta_k;\lambda_k}^{(q)|\bfb} \to \bfQ_{\delta;\lambda}^{(q)|\bfb} $ weakly if $(\delta_k, \lambda_k)\to (\delta, \lambda)$ in the $\|\cdot\|_\ell$ metric. First we will show weak convergence for bounded subsets.
      \begin{proposition}\label{finitevolume}
         Let $ B \subset\subset \bbZ^d \times \bbR$, and let $\bfb = \bff,\bfp,\bfw$, then $\bfQ_{\delta_k;\lambda_k}^{(q)}|_{B}^\bfb \to \csm{\delta}{(q)}|_{B}^\bfb$ weakly.
      \end{proposition}
      \begin{proof}
         It is enough to show the finite dimensional distributions converge
         \cite{billingsley}(Theorem 12.6).
         The finite dimensional distributions in this case are
          events counting the number of bonds and cuts in bounded intervals, that is events of the type,
         \be \label{fdset}
             U^{r_1,\ldots,r_n}_{(z_1;t_1,s_1),\ldots,(z_n;t_n,s_n)} = 
           \{\con:|\con\cap \{z_i\}\times (-t_i,s_i)| = r_i;i=1,\ldots,n\},
         \ee
         where $z_i\in\zd\cup\edge$, $t_i,s_i \in \bbR^+$
         and $r_i\in \zee^+$.
         The difficulty is in the case $q>1$ as the case $q=1$ is simply a product of Poisson distributions.
           Recall, for any $\epsilon>0$ we define $V_\chi$ to be the compact event from Proposition \ref{tightfamily}, so that $|V_\chi| > 1-\epsilon$;
           the spacing of cuts implies for any bounded $B\subset \bbZ^d\times\bbR$, 
           there is some $\cK_\chi<\infty$ so that $k_B(\omega)< \cK_\chi $ for any $\omega\in V_\chi$.

         For $\eta > 0$ and $\chi$ as above, let $f_{\chi,\eta }$ be a continuous function 
         on $V_\chi$ approximating $k_B$. 
         $k_B$ takes on positive integer values and only may change value where a cut \tql moves past\tqr\ a bond,
           that is, $k_B$ is discontinuous at $\omega$ only if there are $x,y,t$ so that $(x,t)\in \omega$ and $(\{x,y\},t)  \in \omega$.
           Therefore for $1/2 > \eta > 0$ we can
          require $f_{\chi,\eta}$ to be bounded by $k_B$ and equal to $k_B$  
          for $\omega$ such that, for any $x,y,t,s$ so that $(x,t) \in \omega$ and $(\{x,y\},s ) \in \omega$
        then $|s-t| > \eta \chi_x(\lfloor t \rfloor) $.

         As $\epsilon,\eta\to 0$,
         by dominated convergence and (\ref{L1}),  
         $I_{\chi, \eta} = \bfQ^{(q)}_{\delta_k;\lambda}|_{B}^\bfb\left(|q^{f_{\chi,\eta}} -q^{k_{B}} |\right)\to 0$.
         In fact, convergence of $I_{\chi,\eta}\to 0$ is uniform in $\delta_k$.
         Indeed, for any $\epsilon$ let $\chi$ be chosen so that
        $   \bfQ_{K;\lambda}^{(q)|\bfb}\left(q^{\cK_\chi}  ; V^c_\chi\right) < \epsilon/2 $.
         Similarly, let $f_{\chi,\eta}$ be chosen so that 
             $\mathcal{O} = \{\omega \in V_\chi:  f_{\chi,\eta} \neq k_B \}$, an open set in $V_\chi$ so that 
         $\bfQ_{K;\lambda}^{(q)|\bfb}  \left(q^{\cK_\chi} ; \cO  \right) < \epsilon/2 $.
         Then, since $V_\chi$ is compact, the set $V_\epsilon = V_\chi\cap\mathcal{O}^c$ is compact,
           and
          $\bfQ_{K;\lambda}^{(q)|\bfb}  \left(q^{\cK_\chi} ; V_\epsilon^c \right) < \epsilon $.
         Thus, for  $f_\chi = k_B$ on $V_\epsilon$ and bounded by $k_B $ 
         on $V_\epsilon^c$, then we have for all $k$
         \[ \bfQ_{\delta_k;\lambda_k}^{(q)|\bfb}\left(|q^{k_B} - q^{f_\chi}|\right)  \leq \bfQ_{K;0}^{(q)|\bfb}\left( |q^{k_B} - q^{f_\chi}|\right)  < 2\epsilon. \]

         Therefore, using weak convergence for such continuous functions,
         $\bfQ_{\delta_k;\lambda}^{(q)}|^\bfb_{B}(g)
           \to \bfQ_{\delta;\lambda}^{(q)}|^\bfb_{B}(g)$
          for $g = q^{f_\chi}$ or $g = q^{f_\chi}{\bf 1}_{U}$,
          for a set $U$ as in (\ref{fdset}).
         Thus,
         \begin{align*}
               & \left|\bfQ_{\delta_k;\lambda}^{(q)}|_{B}(U) - \bfQ_{\delta,\lambda}^{(q)}|_{B}(U) \right|        \\
           & = \left|\frac{\bfQ_{\delta_k;\lambda_k}|_{B}\left({\bf 1}_Uq^{k_{B}}\right)}
                  {\bfQ_{\delta_k;\lambda_k}|_{B}\left(q^{k_{B}}\right)}
            -\frac{\bfQ_{\delta;\lambda}|_{B}\left({\bf 1}_Uq^{k_{B}}\right)}
                  {\bfQ_{\delta;\lambda}|_{B}\left(q^{k_{B}}\right)} \right|  
           \leq 
           \begin{matrix} \left|\frac{\bfQ_{\delta_k;\lambda_k}|_{B}\left({\bf 1}_Uq^{k_{B}}\right)}
                  {\bfQ_{\delta_k;\lambda_k}|_{B}\left(q^{k_{B}}\right)}
                   - \frac{\bfQ_{\delta_k;\lambda_k}{}|_{B}\left({\bf 1}_Uq^{f_\chi}\right)}
                  {\bfQ_{\delta_k;\lambda_k}|_{B}\left(q^{f_\chi}\right)}\right|+
                  \\
                +\left|\frac{\bfQ_{\delta_k;\lambda_k}|_{B}\left({\bf 1}_Uq^{f_\chi}\right)}
                  {\bfQ_{\delta_k;\lambda_k}|_{B}\left(q^{f_\chi}\right)}
                   - \frac{\csm{\delta}{}|_{B}\left({\bf 1}_Uq^{f_\chi}\right)}
                  {\csm{\delta}{}|_{B}\left(q^{f_\chi}\right)}\right| +   
               \\   +\left|\frac{\csm{\delta}{}|_{B}\left({\bf 1}_Uq^{f_\chi}\right)}
                  {\csm{\delta}{}|_{B}\left(q^{f_\chi}\right)}
                  -\frac{\csm{\delta}{}|_{B}\left({\bf 1}_Uq^{k_{B}}\right)}
                  {\csm{\delta}{}|_{B}\left(q^{k_{B}}\right)} \right|  
                  \end{matrix}
         \end{align*}
         The second term vanishes by sending $k\to\infty$, by weak convergence
         of the independent percolation model. The first and third term vanish
         by sending $\epsilon,\eta\to 0$ by the dominated convergence theorem.
      \end{proof}
      
      Let $\delta_k,\lambda_k$ be a sequence such that there exists a non-negative, bounded function $\delta: \bbZ^d \to \bbR$ and  $\lambda \geq 0$  such that 
    \[ \|(\delta_k,\lambda_k) - (\delta,\lambda)\|_\ell \to 0,\]  
    (where $\|\cdot\|_\ell$ is defined in (\ref{localmet})).
      \begin{lemma}\label{lxtconv}
         Let $q \geq 1$ and let $\bfb = \bff,\bfp,\bfw$, then $ \bfQ_{\delta_k;\lambda_k}^{(q)|\bfb}\to \bfQ_{\delta;\lambda}^{(q)|\bfb} $ weakly.
            \end{lemma}
       Here these infinite volume measures are limits   guaranteed to exist by Proposition \ref{weaklim}. 
      This extends the conclusion of  Proposition \ref{finitevolume} to the infinite volume $\lxt$.
      \begin{proof} For $U \in \cF_k$ and for any $n > k$ we have,
         by Proposition \ref{finitevolume}, we have  
         \[\bfQ_{\delta_{k};\lambda_k }^{(q)}|_{B(n)}^\bfb(U)\to \csm{\delta}{(q)}|^\bfb_{B(n)}(U).\]
         We use a diagonalization argument.
        For $i \geq 1$ let $m_i^k$ be chosen so that
         \[ \left|\bfQ_{\delta_{m^k_i};\lambda_{m^k_i}}^{(q)}|^\bfb_{B(k)}(U) - \csm{\delta}{(q)}|^\bfb_{B(k)}(U)\right|<i^{-1}. \]
         Given $\ol n \geq k$ and a sequence $(m_i^{\ol n})$, let $(m_i^{\ol n + 1})$ be a subsequence so that 
         \begin{equation}\label{lxtconv1}
          \left|\bfQ_{\delta_{m^{\ol n+1}_i};\lambda_{m^{\ol n+1}_i}}^{(q)}|^\bfb_{B(\ol n+1)}(U) -
		      \csm{\delta}{(q)}|^\bfb_{B(\ol n+1)}(U)\right|<i^{-1}.
         \end{equation}
         Let $\ell(i) = m_i^i$, and $n > k$ and consider the bound,
         \[
            \left|\bfQ_{\delta_{\ell(i)}  ; \lambda_{\ell(i)}   }^{(q)|\bfb}(U) - \csm{\delta}{(q)|\bfb}(U) \right| 
                        \leq \nonumber 
             \begin{matrix}       
          \left|\bfQ_{\delta_{\ell(i)} ;\lambda_{\ell(i)}  }^{(q)|\bfb}(U) 
               -\bfQ_{\delta_{\ell(i)};\lambda_{\ell(i)} }^{(q)}|^\bfb_{B(n)}(U) \right| +\\  
          +  \left|\bfQ_{\delta_{\ell(i)}  ;\lambda_{\ell(i)}}^{(q)}|^\bfb_{B(n)}(U) - \csm{\delta}{(q)}|^\bfb_{B(n)}(U) \right|\ + \\
          +   \left| \bfQ_{\delta;\lambda}^{(q)}|^\bfb_{B(n)}(U)  - \bfQ_{\delta;\lambda}^{(q)|\bfb}(U)  \right|. 
          \end{matrix}
         \]
            The second term is small for all large $n$, for all large $i$  by (\ref{lxtconv1}). Taking $n$ sufficiently large completes the proof by Proposition \ref{weaklim}. 
      \end{proof}

\subsubsection{Application of continuity to communication events}
 
      As claimed above, communication events are events of continuity, which we demonstrate now.
      \begin{proposition}\label{zerobdry} 
                Let $\delta : \bbZ^d\to \bbR$ be a non-negative bounded function and let $\lambda \geq 0 $. Any communication event $ A = \{ W_1 \llra W_2\}$
        for bounded $W_1$ and $W_2$
        is an event of continuity.  
      \end{proposition}
      \begin{proof}
         First let $ q = 1$.
         We define two events $D_i$, $i=1,2$. The first event $D_1$ is the case that along some line $\{s\}\times\bbR$ 
         for $s\in \bbZ^d\cup \bbE^d$ $\omega$ has an accumulation point of cuts or bonds.
         And the second, $D_2$ is that $\omega$ has a bond and cut which coincide,
         that is there is some time $t\in \bbR$ and pair $\{x,y\}$ so that $(x,t)$
		is  a death and $(\{x,y\},t)$ is a bond in $\omega$.
         Observe that, the boundary of $A$ is contained in the union, $D_1\cup D_2$.

           To see that $\bfQ_{\delta_k,\lambda_k}(D_1) = 0 $ notice this follows from the construction of 
            Proposition \ref{tightfamily}, as for any $\epsilon $ the associated $\chi$
            has $D_1 \subset V_\chi^c$ and  $V_\chi $ has meause $1 - \epsilon$ so the union
              (over $\chi $ chosen for $\epsilon > 0$)  $\cup_\chi V_\chi$  
         has measure 1.
        A similar construction spacing bonds and deaths obtains $\bfQ_{\delta_k,\lambda_k}(D_1) = 0 $, 
         for further discussion, see \cite{bez91}.

         To obtain the statement for $q>1$,
        note $\csm{\delta}{(q)}|_{B}$ 
       has Radon Nikodym derivative 
               $\frac{ q^{k_B(\con)}   }
              {\int_\cons q^{k_B(\con')} d\csm{\delta}{ }(\con') }  $ so it
         is absolutely continuous 
         with respect to $\csm{\delta}{ }|_{B}$; it follows that 
         $\csm{\delta}{(q)}|_{B}\left( D_i \right)$ = 0.  
      \end{proof}
      
       Proposition \ref{weakc} now follows directly from Proposition \ref{zerobdry}         and Lemma \ref{lxtconv}.
      
     It follows from Proposition \ref{weaklim}  and the FKG inequality that, for any collection of pairs of bounded sets $W_1^{(j)}, W_2^{(j)}$ for $j = 1,..,n$, so that for $\bfb = \bff,\bfp,\bfw $
 \be\label{fkginfty}
    \bfQ_{\delta;\lambda}^{(q)|\bfb}\left( \bigcap_{j = 1,..,n} \{ W_1^{(j)} \lra W_2^{(j)} \} \right) \geq \prod_{j = 1,..,n}   \bfQ_{\delta;\lambda}^{(q)|\bfb}\left(  W_1^{(j)}\lra W_2^{(j)} \} \right). 
 \ee

    \begin{proof}[proof of Proposition \ref{mreg}]
     For $(h_k,\lambda_k,\theta_k) \to (h,\lambda,\theta)$ in  $(\cC^+(\bfTh)\times \bbR^+ \times \bfTh   ,  \wh \rho)$
     we have,  for each $x\in \bbZ^d$ that
     $h_k(\bfT^x \theta_k) \to h(\bfT^x \theta)$ so that
    $\wh M^{(q)|n}_{h_k,\theta_k;\lambda_k}  \to  \wh M^{(q)|n}_{ h,\theta;\lambda} $ by Propostition \ref{weakc}.
     Thus, as $\wh M^{(q)|n}$ is continuous in $(h,\theta,\lambda)$, in the $d_\infty^\natural$ metric, and  
      $\wh M^{(q)|n}$   is decreasing in $n$,
    it follows that $\wh M^{(q)|n}_{h,\theta;\lambda}$ is upper semi-continuous in $d_\infty^\natural$. The statement $\mu(\wh\cU_0) \in \{0,1\}$ follows from Proposition \ref{regs}. 
    \end{proof}

    \subsection{The 0 - 1 proof}\label{reg2}

   \begin{proposition}\label{regs}
     For any $h$, $\lambda$ the set $\wh \cU_0  $ has full measure or zero measure. 
     If one of the maps $T_1,..,T_d$ is uniquely ergodic, then $\mu(\wh\cU_0 )= 1$  or there is some $\epsilon > 0$ so that 
     $ \wh M_{q,h,\lambda}(\theta )  > \epsilon $ for every $\theta \in \bfTh$.
   \end{proposition}

     \begin{proof}
  
 Suppose $\{\theta\in \bfTh:\wh M_{h,\theta;\lambda} >0\}$ has positive ergodic measure $\mu$. As $\bfT$ is ergodic, for any $\theta_0$ there is some $ x$ so that $\wh M_{h,\bfT^x\theta_0;\lambda} > 0 $. Let $P_x =\{0 = w_0,w_1,..,x= w_n\}$ be a path of nearest neighbor steps in $\bbZ^d$ from $0$ to $x$. Let $D_{x}$ be the event that there are no deaths on $P_{x} \times [0,1]$ and there is at least one bond on each nearest neighbor step $\{w_{i-1}, w_i\}\times [0,1]$ for $i = 1,..,n$. Clearly $\bfQ^{(q)}_{h,\theta;\lambda}( D_x) > 0$ so, by the FKG inequality,
 \[ \wh M_{h,\bfT^x\theta_0;\lambda}^{(2)} \geq 
      \bfQ^{(q)}_{h,\theta;\lambda} (D_x) \wh M_{h,\theta_0;\lambda}^{(2)}>0.    \]

    Finally, suppose $T_i$ is uniquely ergodic. Let $\wh \cU_\epsilon = \{ \theta :\wh M^{(q)}_{h,\theta;\lambda} > \epsilon\}$,  then if $ \mu( \wh \cU_\epsilon ) > 0$ there is some $N$ so that for any $\theta$  there is some $m$ so that $0 < m \leq N $ and 
     $ T_i^{m }\theta  \in \wh \cU_\epsilon$.  
    We therefore have for all $\theta$,
    \[    \wh M^{(q)}_{h,\theta;\lambda}   \geq 
    \epsilon
     \bfQ^{(q)}_{h,\theta;\lambda}( D_{Ne_i}  ).
                     \]
         where $e_i\in \bbR^d$ is the vector with 1 at index $i$ and zeros at all other indices. 
      \end{proof}

   \subsection{The $G_\delta$ construction}\label{Borel}

      In this section we will prove Propostion \ref{gdeltatop}.
   
      Let $X$ be a space with a Borel topology and probability measure $\nu$. Let $Y$ be a space with  a Borel topology.
     Let $V$ be a map from $X $ to the Borel sets of  $Y$,
     and let $W$ be an `inverse' i.e. $W_y = \{ x\in X: y \in V_x \}$.
     Let $J_\eta = \{y : \nu(W_y) > \eta\}$.
         \begin{lemma}\label{sample2measure}
     If $V_x$ is open for all $x \in X$ then $J_\eta $ is open for all $\eta >0$.
    \end{lemma}
    \begin{proof}
      Suppose $y_i \in J^c_\eta $ and $y_i \to y$.
      For any  $x$ so that $y\in V_x$, as $V_x$ is open and $y_i \to y$, we have $y_i \in V_x$ for all large $i$. Thus, $x \in W_{y_i}$ for all large $i$, as this holds for all $x\in W_y$, it implies $\liminf W_{y_i} \supset W_y $.
        But for all $i$, $y_i \in J_\eta^c$ so $ \nu (W_y )\leq \nu (\liminf W_{y_i} ) \leq \eta $ and therefore $ y \in J_\eta^c$.  
    \end{proof}

 We apply this lemma with $X=\bfTh $ and $Y = \cC^+_{fin}(\bfTh)$. Instead of fixing the scaling of the sampling function here, it is simpler to work with fixed $\lambda = 1$. Before proceeding to the proof of Proposition \ref{gdeltatop}, we show the inclusion map $\iota$ is open.

     Observe that the identity map $\iota  :  ( \cC^+(\bfTh),d_\infty ) \to ( \cC^+(\bfTh), d_\tau )$ is an open continuous map. Continuity is clear, to see $\iota $ is open, 
    we will show the image of the set   $ B_\infty(h,\epsilon) =\{h' \in  \cC^+(\bfTh) : d_\infty(h' , h) < \epsilon  \}$
    under $\iota$   is open. 
     If  $\epsilon  >  \|h\|_\infty  $ the image of  $ B_\infty(h,\epsilon) $ is all normalized functions, i.e. $ \tau ( B_\infty(h,\epsilon) ) = \tau (  \cC^+(\bfTh) ) $.
     Thus, it is sufficient to consider $\epsilon \leq \|h\|_\infty$. 
     
     Now consider $f \in \tau( B_\infty(h,\epsilon) )$ and let $h' \in B_{\infty}(h,\epsilon)$ so that $\tau (h') = f$. Let us write $g = h' -  h$ where $\|g\|_\infty < \epsilon$. Now suppose $\delta <   \frac{\epsilon - \|g\|_\infty }{ \|h\|_\infty + \|g\|_\infty}$ and suppose $h''\in \cS_\cF$ so that $d_\tau (h',h'') < \delta$. Let $h^* = \|h'\|_\infty\tau(h'')$
    then we have 
     \[\|h^* - h\|_\infty  = \|g + \|h'\|_\infty   \tau(h'') - h'\|_\infty      \leq  \|g\|_\infty + \|h'\|_\infty |\tau(h'') - \tau(h')| < \epsilon.   \]
   Thus $B_\tau(h',\delta) = \{h'' \in  \cC^+ (\bfTh) :
    d_\tau(h' , h'') < \delta  \} 
    \subset \tau (B_\infty(h,\epsilon ))$, from which we have that $\iota $ is an open mapping.

     \begin{proof}[proof of Proposition \ref{gdeltatop}]    
              Let us fix $\lambda = 1$.
         Let $V_\theta^{(\epsilon)} = \{ h \in  \cC^+_{fin}(\bfTh) : \wh M^{(q)}_{h,\theta;\lambda} < \epsilon\}$, 
        as  $\wh M^{(q)}$ is upper semi continuous in $d_\infty $, we have  $V_\theta^{(\epsilon)}$ is open.
              Let $W_h^{(\epsilon)} = \{\theta \in \bfTh: h \in V_\theta^{(\epsilon)}\}$
              then, $J_{\eta}^{(\epsilon)} =  \{ h \in  \cC_{fin}^+(\bfTh) :  W_h^{(\epsilon)}    >  \eta \}$ is open by Proposition \ref{sample2measure}.
              As discussed above,   the inclusion $(\cC^+_{fin}(\bfTh), d_\infty) \to (\cC_{fin}^+(\bfTh), d_\tau)$ maps open sets to open sets so $J_{\eta}^{(\epsilon)}  $ is open in $(\cC_{fin}^+(\bfTh), d_\tau)$, moreover, $ \wh J_{\eta}^{(\epsilon)} =  \tau^{-1}\tau J_{\eta}^{(\epsilon)} $ is open in $(\cC_{fin}^+(\bfTh), d_\tau)$.
            Let $\epsilon_i \searrow 0 ,\eta_i \nearrow 1$ then    $ \cap_i  \wh J_{\eta_i}^{(\epsilon_i)} $ is a $G_\delta$ in $(\cC_{fin}^+(\bfTh), d_\tau)$.
             Finally   $\cap_i  \wh J_{\eta_i}^{(\epsilon_i)}  $  is the set of $h$ 
             so that $\wh M_{q,h,1}(\theta) = 0$ for $\mu$ almost every  $\theta$.
             To check that  $\cap_i  \wh J_{\eta_i}^{(\epsilon_i)} = \wh F_{weak} $, consider $g \in \wh F_{weak}$, let $\lambda > 0$ be small enough that $\wh M^{(q)}_{g,\lambda}(\theta) = 0$ for $\mu$ almost every  $\theta$.
             As $\wh M^{(q)}$ is invariant under scaling, $\wh M^{(q)}_{g/\lambda,1}(\theta) = 0$ for $\mu$ almost every  $\theta$,
          so $g/\lambda \in J_{\eta_i}^{(\epsilon_i)}$ for all $i$.   But $\tau(g/\lambda) = \tau(g)$ so $g \in \cap_i  \wh J_{\eta_i}^{(\epsilon_i)} $.

      \end{proof}

         To prove the remark,
         let $\wh V_\theta^{(\epsilon)} = \{(h,\lambda) \in \cC_{fin}^+(\bfTh) \times \bbR^+ : \wh M_{q,h,\lambda}(\theta)  < \epsilon\}$,  again this set is open.
           The rest of the proof is similar to the first part by setting 
              $\wh W_{h,\lambda}^{(\epsilon)} = \{\theta \in \bfTh: (h,\lambda) \in \wh V_\theta^{(\epsilon)}\}$
              and $\wh J_{\eta}^{(\epsilon)} =  \{ (h,\lambda) \in  \cC^+_{fin}(\bfTh)  :\wh  W_h^{(\epsilon)}    >  \eta \}$.


  \subsection{Proof of Theorem \ref{introtheorem}} \label{Secmainproof}
 
   Let us begin by comparing the magnetism parameter $M_{\cdot,\cdot;\cdot}$ to the percolation parameter $\wh M^{(2)}_{\cdot,\cdot;\cdot}$.

   \begin{proposition} \label{mmhat compare}
      Let $h \in \cC^+_{fin}(\bfTh)$ be fixed. For any $\theta \in \bfTh$, and $\lambda > 0$, 
       \be\label{mmh}
       \wh M^{(2)}_{h,\theta;\lambda} \geq M^{(2)}_{h,\theta;\lambda}   \geq 
               \wh M^{(2)}_{h,\theta;\lambda}   \| \wh M^{(2)}_{h,\cdot;\lambda} \|_\infty.  
               \ee
     Moreover, $\mu(\cU_0) \in \{0,1\}$, and, if some $T_i$ is uniquely ergodic, then $\mu(\cU_0(h,\lambda)) = 1$ or there is some $\epsilon > 0 $ so that, for all $\theta\in \bfTh$, $ M^{(2)}_{h,\theta;\lambda}  > \epsilon $.
   \end{proposition}
   It follows immediately from the first claim that     $F_{\sharp} = \wh F_{\sharp} $ for $\sharp = \{strong,weak\}$.
      Theorem \ref{introtheorem2} follows immediately from the second claim.

     \begin{proof}
        From (\ref{qim rcm rel}) we have $\wh M^{(2)}_{h,\theta;\lambda} \geq M_{h,\theta;\lambda} $. On the other hand, with  $\bfQ^{(q)}_{h,\theta,\lambda}$  probability one there exists at most 1 infinite cluster, this is a standard fact in many percolation models, 
       for the result in the context of iid $\delta$ see \cite{AizKleinNewman}, a proof for ergodic $\delta$ is similar.
         Thus, by the FKG inequality, 
           \[\bfQ^{(q)}_{h,\theta,\lambda}\{(x,0) \lra  (y,0)\} \geq \wh M^{(2)}_{h,\lambda}(\bfT^x\theta) \wh M^{(2)}_{h,\lambda}(\bfT^y\theta).  \]
         But, again by (\ref{qim rcm rel}), the left hand side is $\langle \sigma_x^{(3)} \sigma_y^{(3)} \rangle$.
            Thus, by choosing a sequence of $y_i$ moving to infinity so that
           $\wh M^{(2)}_{h,\bfT^{y_i}\theta;\lambda} \to \sup_{\theta}\wh M^{(2)}_{h,\theta;\lambda} $,   
               we have (\ref{mmh}).
 
 The second claim follows by combining the equivalence of $M$ and $\wh M$ with the result of proposition \ref{regs}.
 \end{proof}

      \begin{proof}[Proof of Theorem \ref{introtheorem}.]
         Given a sampling function $h\in\cC^+_{fin}(\bfTh) $ we find elements of $F_{strong}$ and $F_{weak}$, close to $h$. Let us write $A  = h^{-1} (0)$, the nonempty and finite zero set of $h$, moreover,  without loss of generality we may assume $\|h\|_\infty = 1$. For any $\epsilon > 0$ the constructed  function $h_\epsilon$ respectively in  $F_{strong} $ and $F_{weak}$, such that $d_\infty(h, h_\epsilon) < \epsilon$, will obey $\|h_\epsilon\| = 1$ so that $d_\tau(h,h_\epsilon) = d_\infty(h, h_\epsilon) < \epsilon$.

   First we will find a function in $F_{strong}$ close to $h$. Let $\theta_0 \in A$ and let $\eta > 0 $ be so that $r(\theta_0,\theta) < \eta$ implies $h(\theta) < \epsilon$. Let $f_{\theta_0}$ be the function declared to exist in Theorem \ref{strongfun} so that $f_{\theta_0}(\theta) = 1$ for $r(\theta,\theta_0) > \eta$. Set $h' = f_{\theta_0}h$, then $d_\tau(h',h) < \epsilon$, and by monotonicity $h' \in F_{strong}$.
        
       Now we prove the density  of $F_{weak} $, again let $h \in \cC_{fin,1}^+(\bfTh)$ and write $A = h^{-1}(0)$. Let $\nu $ and $\chi$ satisfy  $1 + \chi < \nu(1 + 1/d)$ and let $ m < \infty$ then let $\psi_A$ be the sampling function introduced in Proposition \ref{fasterthanpolynomial}  so that the pair $(\psi_A, \lambda)$ is $\bfC_{m,\nu}$-localized. By the scaling property, $(t\psi_A,t\lambda)$ is also $\bfC_{m,\mu}$-localized. By monotonicity, $ ( h +  t\psi_A, t\lambda )$ is also $\bfC_{m,\nu}$-localized. Let $h_t = h + t\psi_A$, then  $ \|h - h_t\|_\infty \to 0 $ as $t \to 0$. It follows that  $d_\tau(h,h_t) \to 0 $ as $t \to 0$ which establishes  density of $F_{weak} $ in $(\cC^+_{fin,1}, d_\tau)$.
 
 Proposition \ref{mmhat compare} shows that $F_{weak} = \wh F_{weak}$ and Proposition \ref{gdeltatop} states $\wh F_{weak}$ is a $G_\delta$. Thus we have the $G_\delta$ claim of the Theorem. Finally we prove the claim of the partition. Observe that
 \[ F_{strong }  = \cap_{\lambda >0} \{h: \mu( \{\theta\in \bfTh: M_{h,\theta;\lambda} >0\})=1 \} \] 
 If there is some $\lambda > 0$ so that 
 \[\mu( \{\theta\in \bfTh: M_{h,\theta;\lambda} >0\}) < 1 \]
 then, by the second part of Proposition \ref{mmhat compare}, $ \mu( \{\theta\in \bfTh: M_{h,\theta;\lambda} >0\})=0  $, so $h \in F_{weak}$. 
      \end{proof}

          \begin{proof}[Proof of Theorem \ref{fasterthanpolynomial}]
          Let $\chi>0$ be chosen so that $ 1 + \chi < \nu(1 + 1/d)$. For all $k$, let $\eta_k$ be defined as 
  \[\eta_k =   \inf \{ \eta : K_1(\eta) < k  \}.  \]    
  Define $\wt K $ on $(0,1)$ as a continuous function by setting $\wt K(\eta )  = 1$ and for all $k > k_0$, for some $k_0$ so that $K_1(\eta) > k_0$ set $\wt K(\eta_k) = k$, finally, linearly interpolating between these values.    
           Let $g_A(\theta) = 1$ for all $\theta$ so that $r(\theta,A) > \eta$. For $\theta_0 \in A$ and $r(\theta_0,\theta) < \eta $,
          \[ g_A(\theta) = e^{ 1- [\wt K(r(\theta_0,\theta)) ]^\chi} \]
     By construction, $g_A^{-1}(0) = A$ and $ \| g_A  \|_\infty  =  1$. Moreover, $g_A $ satisfies the conditions of Theorem \ref{lowRecurrence}. For $\theta$ so that $\epsilon< r(\theta,A) < \eta $, we have 
     \[ \log(1 + |\log g_A(\theta)|) 
     = \chi \log \wt K(r(\theta,\theta_0)) 
     \leq \chi \log (1 +  K_1(r(\theta,\theta_0))) 
     \leq \chi \log (1 +  K_1(\epsilon ))   \]
          which is sufficient to  verify (\ref{psilimit}) and obtain the proof.
      \end{proof}

       \begin{proof}[Proof of Theorem \ref{strongfun}]
         Let $\eta > 0$ and let $\eta > \eta_0 > \eta_1 >   \eta_2 > \cdots$, so that  $\eta_i \searrow 0$ and for all $i$, the $\eta_i$ ball centered at $\theta_0$ has $\mu$ measure 0 boundary.

         Now, as  $\spc$ is compact for every $i$
         there is some $L_i$ so that 
         \[ \spc \subset \bigcup_{x\in [0,L_i]^d} \ag^x B(\eta_{i+1};\theta_0) ,\]
           note the statement holds if $\spc$ is not compact but some   $T_i$ is uniquely ergodic with respect to $\mu$.
         For $a>1$ we construct the following function.
         Let $v(x) = 1$
         on $1 \geq x > \eta$ and 
         \[  v (x) =   e^{ \left(   \frac{ x - \eta_0}{\eta - \eta_0} +
              \frac{\eta - x}{\eta-\eta_0} L_0   \right)^a }\]
         on $\eta\geq x > \eta_0$ and
         \[  v(x) =  e^{\left(  \frac{ x - \eta_{i+1}}{\eta_{i} - \eta_{i+1} } L_i  +
              \frac{\eta_i - x}{\eta_i-\eta_{i+1}}   L_{i+1}  \right)^a  } \]
         on $\eta_i \geq x > \eta_{i+1}$, for $i \geq 0$.
         And let
         \[ f_{\theta_0}(\theta)=\frac{1}{v(r(\theta,\theta_0))} \]
         which is the desired function. 
         
         To show $f_{\theta_0}$ satisfies (\ref{psilimit2}), let $0 <\epsilon < \eta_0$. Suppose $\eta_i \leq \epsilon < \eta_{i+1} $, then for $r(\theta,\theta_0) < \epsilon$,
         \[ \log( 1 + | \log (f_{\theta_0})| )> a \log L_i.  \]
         By construction, $K_2(\theta_0,\epsilon ) \leq  K_2(\theta_0, \eta_{i+1}) \leq L_i$, thus, $f_{\theta_0}$ satisfies Theorem \ref{hirec}. 
          \end{proof}

   \section{Percolation arguments}\label{analysis}

     We will cover Theorems \ref{rotationtheorem}, \ref{lowRecurrence}, and \ref{hirec} in this section. We begin in Section \ref{multscale} with the introduction of the multiscale analysis statements which provide conditions on the family of environments. In Section \ref{low} we show the environment process and sampling function satisfying (\ref{psilimit}) is sufficient to apply the multscale analysis, which proves Theorem \ref{lowRecurrence}. In Section \ref{ooc} we prove Theorem \ref{hirec}, by coupling the percolation in $\bbZ^d \times \bbR$ to a percolation in $\bbZ^{d+1}$. We complete this section with the proof of Theorem \ref{rotationtheorem}
 in Section \ref{rotations}.

   For spatial length scale $L$ we associate a `time' scale $T(L) = \exp(L^\tau)$, where we will specify $\tau$ depending on the argument. Let us write a cylinder centered at $(x,t)$ as
       \[ B_L(x,t) = \Lambda_L(x)\times\left[t-T(L),t+T(L)\right] \]   
       The environment
      is invariant in the time dimension so the choice of $t$ above will often not affect the discussion,
      thus we fix the notation $B_L(x) = B_L(x,0)$.

      \subsection{Multiscale Analysis}\label{multiscale}
      By the ordering Lemma \ref{RCMordering} we need only demonstrate localization
      for the product measure $q = 1$ to infer similar results for $q > 1$. 
       Thus for this section we fix $q  = 1$.
      \begin{definition}
         Let $m>0$ and $L\in{\bf Z}^+$. 
         For a fixed environment $\delta,\lambda$; a site $x\in\zd$ is $(m,L)$-regular if
         \[ \bfQ_{\delta;\lambda}\left(x \lra \partial B_L(x)\right)\leq \exp\{-mL\} \] 
         otherwise it is $(m,L)$-singular.
         A set $A\subset\zd$ is $(m,L)$-regular if every $y\in A$ is
         $(m,L)$-regular. Otherwise it is $(m,L)$-singular.
      \end{definition}
      \begin{definition}
         A site $x\in \zd$ is $\epsilon$-resonant if $\delta(x)<\epsilon$. 
         A set $A\subset\zd$ is $\epsilon$-resonant if there exists 
         $x\in A$ which is $\epsilon$-resonant.
      \end{definition}
      \begin{definition}
         The pair $(\epsilon,L)$ is $m$-simple if $x\in\zd$ is 
         $(m,L)$-singular implies $\Lambda_L(x)$ 
         is $\epsilon$-resonant.
      \end{definition} 

   From the assumptions in Proposition \ref{lowRecurrence} we have $\nu( d+1) -  (1+ \chi) d > 0$. We will introduce several more parameters for the multiscale analysis.
   Let $\alpha > d $ be a parameter   satisfying,
   \[          \alpha  >  \frac{(\chi + \nu) d}{\nu(d+1) -(1+\chi) d} .   \]
   It follows that 
   \[      0 <  \chi(\alpha + 1) d < \alpha\nu - (\nu  + \alpha - \alpha \nu)d    \]
   Let $\gamma$ and $\kappa$ be parameters so that $\alpha \nu /d  -   \chi(\alpha + 1)> \kappa > \nu + \alpha - \alpha \nu$ and
   \[   \chi(\alpha + 1) d <\gamma  < \alpha\nu - \kappa d    \]
    Finally, let $\tau$ satisfy $\nu < \tau < \kappa - \alpha(1 - \nu)$. In the statements below, $R$ is a fixed positive integer. In practice, for dynamically defined environments, $R = |h^{-1}(0)|$, which we will justify below.

     As the name suggests, multiscale analysis involvles induction of regularity on a sequence of lenght scales.
 
  The following proposition, which is essentially the statements of Sublemmas 4.2 and 4.3 of \cite{klein94}, states regularity on scale $L$ may be upgraded to regularity on scale $L^\alpha$
      \begin{proposition}\label{msa}
     Suppose
     $ \Lambda_{L^{\alpha}}(x) \cap (\cup_{i= 1}^R  \Lambda_{L^\kappa}(y_i)  )$ for a sequence $y_1,..,y_R\in\bbZ^d$  is $e^{  -  L^{\gamma} }$ nonresonant. Suppose every  $y \in \Lambda_{L^\alpha}(x)\less \cup_{i=1}^R \Lambda_{2L+1}(y_i)$ is $(m,L)$-regular.
         Then $x$ is $(m-L^{-\tau},L^\alpha)$ regular.
      \end{proposition}

      Let us note that we can initialize the multiscale analysis with any $m_1$ and large $L_1$ by selecting $\lambda>0$ 
     sufficiently small.
     Let $\delta' = \min_{x \in \Lambda_L(0)}  h(\bfT^x \theta)$ then, for $\lambda / \delta'$ sufficiently small,
      \begin{equation} \label{initialscale}
           \bfQ_{h,\theta;\lambda} (0\llra \partial B_L(0) ) \leq \bfQ_{\delta',\lambda} (0\llra \partial B_L(0) )<  e^{- m_1 L_1}.   
    \end{equation}  
   This inequality follows from Corollary 2.2 in \cite{klein94}.

      The following theorem is the standard use of the Borel-Cantelli lemma in 
      multiscale arguments \cite{cam91}, \cite{klein94}. 
      It is also stated as Theorem 2.1 in \cite{jito} in a more general form.
       The sequences $(m_i)$ and $(L_i)$ correspond to the sequences generated by applications of Proposition \ref{msa} and 
          are defined for   initial $m_1$ and $L_1$ 
       and induction $m_{i+1} = m_i - L_i^{-\kappa} $ and $L_{i+1} = L_i^\alpha$.

        For given sampling function $h$ and coupling $\lambda$
      let us write 
         \be   a_{k} = \mu(\left\{\theta:  \textrm{Environment initialized  at }\theta 
            \textrm{  is } \left(m_k,L_k\right)\textrm{-singular}\right\})
        \label{ak}            
              \ee
      If the sequence $a_k$ decays sufficiently fast, we can apply the following theorem, 
      which follows from Theorem 3.3 in \cite{klein94}.
      \begin{theorem}\label{msaBC}
         Fix coupling $\lambda > 0$ and sampling function $h$. Let $p > \alpha d$,
         then if $\limsup  a_k L_k^p  < \infty $ we have that 
         for any  $0<m<m_\infty = \inf_k m_k$, the pair $h,\lambda$ is $C(\nu, m )$-localized.
      \end{theorem}

\subsection{Recurrence}\label{low}

   We carry out the generalization of the arguments in \cite{jito} for
   the specified controlled recurrence models.
    First we state bound for the probability  $\Lambda_L$ is $\epsilon$-resonant, we use the following transversality notation for sampling functions.
    
     For an increasing function $Z:(0,1) \to (0,1)$ so that $\lim_{r \to 0} Z(r) = 0$, we say $h \in \cC_{fin}^+(\bfTh) $ is admitted by $Z$, if
     \[   \lim_{\epsilon \to 0}\left( \inf_{\theta: r(\theta, h^{-1}(0) ) > \epsilon} \frac{Z(h(\theta))}{ r(\theta, h^{-1}(0) )} \right) > 3 . \]
     The choice of the constant $3$ here is somewhat arbitrary, but it makes the proof of Proposition \ref{KacsLemma} and the following proofs more convenient.
     
   \begin{proposition}\label{KacsLemma} 
      Let $(\bfTh, \bfT)$ be an environment process, let
  $ h \in \cC_{fin}^+(\bfTh)$, and let $Z$ be a function admitting $h$. For sufficiently small $\epsilon$,
 \[  \mu\left( \Lambda_L(0) \textnormal{ is } \epsilon \textnormal{-resonant} \right) \leq   \frac{ c L^d   }{ K_1(  Z(\epsilon)) } , \]
 where $0 < c< \infty$ depends only on $d$ and  $|h^{-1}(0)|$. 
   \end{proposition} 
   \begin{proof}
      We assume that $h$ is admitted by $Z$ and $\epsilon $ is sufficiently small that  $r(\theta,h^{-1}(0)) < \epsilon$ implies $Z(h(\theta)) >2 r(\theta, h^{-1}(0))$.
      Let us write $ A =  \{\theta_1,\ldots, \theta_{|A|}\}$, 
         and $U_i = h^{-1}([0,\epsilon))\cap B_{Z(\epsilon)}\left(\theta_i\right)  $
        so that
      \begin{equation}  \label{partiKac}  h^{-1}([0,\epsilon) )= \bigcup_{1\leq i\leq |F|} U_i. \end{equation}
      We consider each $U_i$ separately, and we will show $K_1(Z(\epsilon))$ is a lower bound on return times to the set $U_i$. Indeed, suppose there is some $\theta$, so that $x,y\in\bbZ^d$ are such that $\bfT^x\theta,\bfT^y\theta\in U_i $. Then, for some $0 < c < \infty$
      \[r(\bfT^x\theta,\bfT^y\theta) \leq  r(\bfT^x\theta, \theta_0) + r( \theta_0,\bfT^y\theta) \leq  \tfrac12 Z(h(\bfT^x\theta)) + \tfrac12 Z(h(\bfT^y\theta)) \leq   Z(\epsilon). \] 
      Where the last inequality holds because $h(\bfT^x\theta),h(\bfT^y\theta) < \epsilon$.
      By definition of $K_1$,
      this implies that $K_1(Z(\epsilon)) \leq |x-y| $, ie $K_1\circ Z$ is a lower bound on return times.
      On the other hand, the return times are related to the size
       of the set in the ergodic measure by Kac's lemma, (see e.g. \cite{cfs1982}) thus,
      \begin{equation}\label{Kac}
    K_1(  Z(\epsilon)) \leq   \bbE\{\textrm{Return time to } U_i  \} 
          = \frac{1}{\mu\left(A_i\right) }.
      \end{equation} 
     Thus, combining (\ref{partiKac}) and (\ref{Kac}),
      \[\mu(h^{-1}([0,\epsilon))) \leq 
             \sum_i \mu\left(U_i\right)  \leq \frac{|A|}{ K_1(  Z(\epsilon))}. \]
        Finally, the conclusion follows from the fact that the probability of a set $\Lambda\subset \bbZ^d$ being $\epsilon$-resonant is
        bounded by $|\Lambda|\cdot \mu(h^{-1}([0,\epsilon))).$  
   \end{proof}
   We show that simplicity at scale $L$ implies regularity in the bulk at scale $L^\alpha$. 
   Essentially, we show that at the chosen sequence of scales, 
    at most one resonance occurs per $L_i$ box per  zero of $h$. 
   \begin{proposition}\label{simple2regular}  
      Let $(\bfTh, \bfT)$ be an environment process.
      Suppose $h\in \cC_{fin}^+(\bfTh)$ is admitted by a function $Z$ such that, for small enough $\epsilon > 0$, 
      \[     
        K_1( Z(\epsilon)) >  \log^{1/\chi} \left(\epsilon^{-1}\right)
      \]
        for some $\chi < 1/d$ and let $A = h^{-1}(0)$. Moreover, suppose $(\exp\{-L^\gamma\},L)$ is $m$-simple and $L$ is large. Then for any $x\in\zd$
      there exists $y_i\in\Lambda_{L^{\alpha}}$, for $i=1,..,|A|$, so that  $\Lambda_{L^\alpha}(x)\less \cup_{i=1}^{|A|} \Lambda_{L}(y_i)$ is $(m,L)$-regular.
   \end{proposition}
   \begin{proof}
      Fix an initial phase $\theta$ and label the points of $A$ by $A = \{\theta_1,\ldots,\theta_{|A|}\}$. If $x\in\bbZ^d$ is $\epsilon$-resonant then for some $i = 1,..,|A|$, $r(\bfT^x\theta,\theta_i) <\tfrac12 Z(\epsilon)$ for some $i=1,..,|A|$. If $\bfT^x\theta,\bfT^y\theta  \in B_{Z(\epsilon)}(\theta_i)$ then  $|x-y| \geq  K_1(Z(\epsilon))$, as in the proof of Proposition \ref{Kac}.       
      Let $\epsilon = \exp\{-L^\gamma\}$, then from the assumption $h$ we have, for $L$ sufficiently large, $ |x-y| > L^{\gamma/\chi}$, if both $h(\bfT^x\theta),h(\bfT^y\theta) < \epsilon$. 
      Thus for each $\theta_i\in h^{-1}(0)$ there exists at most one $ \exp\{-L^\gamma\} $-resonant
      $y_i\in\Lambda_{L^\alpha}(x)$ . By definition of $m$-simple, the result follows.
   \end{proof}

   \begin{proposition}\label{simple2simple} Let $(\bfTh, \bfT)$ be an environment process.
      Suppose $h\in \cC_{fin}^+(\bfTh)$ is admitted by a function $Z$ such that, for some $ 0 < c < \infty $ and for small enough $\epsilon > 0$, 
      \[     
        K_1( Z(\epsilon)) >  \log^{1/\chi} \left(\epsilon^{-1}\right)
      \]
        for some $\chi < 1/d$. Suppose $\lambda $ is sufficiently small and $(\exp\{-L^\gamma\},L)$ is $m$-simple, then $(\exp\{-L^{\alpha \gamma}\},L^\alpha)$ is $m' = m - L^{-\tau}$ simple.
   \end{proposition}
   \begin{proof} Let  $A = h^{-1}(0)$, by Proposition \ref{simple2regular}, we have that there exists $y_1,\ldots,y_{|A|}\in\Lambda_{L^{\alpha}}$ so that $\Lambda_{L^\alpha}(x)\less \cup_{i=1}^{|A|} \Lambda_{L}(y_i)$ is $(m,L)$-regular.
      Now if $x\in\zd$ is so that $\Lambda_{L^{\alpha\gamma}}$ is 
      $\exp\left\{-L^{\alpha\gamma}\right\}$ non-resonant,
      Theorem \ref{msa} implies $x$ is $(m',L^\alpha)$ regular.
   \end{proof}

   \begin{proof}[Proof of Theorem \ref{lowRecurrence}.]   
      To prove the theorem we need to show the hypothesis of theorem \ref{msaBC} holds.
     Let $\gamma,\alpha,\tau,\kappa$  satisfy the assumptions of Proposition \ref{msa} 
    so that $\chi (\alpha + 1) d < \gamma$.  
     Let $p > \alpha d$ be such that $\chi ( p + d) < \gamma $.
 
   By definition of $\phi_h$, $Z$ is admitted by $ \phi^{-1}$, so that by (\ref{psilimit}), for small enough $\epsilon$,
   \be \label{psilim3}
    \frac{\log|\log\epsilon|}{\log K_1 \phi^{-1}_h(\epsilon)  } < \chi.        
   \ee
   Then by Proposition \ref{KacsLemma}, setting $\epsilon = e^{-L^\gamma}$
    \be \label{reslevel}  \mu(\Lambda_L(0) \textnormal{ is } e^{- L^\gamma } \textnormal{-resonant}) \leq  L^{d - \gamma/\chi}.  \ee
 We will require $L_0$ to be large enough that   for all $L \geq  L_0$ (\ref{reslevel}) holds.

     Let $m $ be as defined in Theorem  \ref{lowRecurrence}, let $m_0 = m+ 1$. For chosen $L_0$ we have a sequence of scales $ L_{k+1} = L_k^\alpha$, $\epsilon_k = \exp\left\{-L_k^\gamma\right\}$ and $m_{k+1} = m_k - L_k^{-\tau}$. Let $L_0$ be large enough that  $ m <  m_\infty = m_0 - \sum_{i=0}^\infty   L_i^{-\tau}$.

   Let $L_0$ be sufficiently large and        take bond rate $\lambda > 0$ so small that  the uniform environment $\lambda$ and $\delta = \epsilon_0$  has probability of escape 
    \[   \bfQ_{\epsilon_0,\lambda} (0 \lra \partial B_{L_0}(0)  )< e^{-m_0 L_0} .  \] 
    Thus in the disordered environment, by comparison to the homogeneous environment using (\ref{order1}),
     $(\exp(-L_0^\gamma),L_0)$ is $m_0$-simple.

      Apply Proposition \ref{simple2simple}, using (\ref{psilimit}), $Z = \phi^{-1}_h$, and large enough $L_0$, we have for all $k$, $(\epsilon_k,L_k)$ is $m_k$ simple. Now to complete the proof we only need to check that the $a_k$ as defined in (\ref{ak}) decays sufficiently fast. By (\ref{reslevel}), $a_k < L_k^{d-\gamma/\chi}  $ it is indeed true that $\limsup_k a_k L_k^{p} < \infty$, which satisfies the hypothesis of Theorem \ref{msaBC}.
   \end{proof}


\subsection{Infinite components}\label{ooc}
      In this section we summarize results from \cite{AizKleinNewman} as they apply 
      to our model. 
     First we state the uniqueness of the infinite cluster, and an immediate corollary, 
      a lower bound on the probability two sites communicate.

      \begin{theorem}
         Let $\delta(x) = h(\ag^x\theta)$.
         There exists with probability one either zero or one unbounded components.
         Furthermore, for $x,y\in\bbZ^d$ and any $t,s\in\bbR$, 
         \[ \csm{h,\theta}{}\left( \{(x,t)\lra (y,s)\}\right) \geq M_\lambda( \bfT^x\theta )M_\lambda(\bfT^y\theta). \] 
      \end{theorem}

     The uniqueness of the infinite cluster is shown in \cite{AizKleinNewman} for random independent environments. 
      In \cite{AizKleinNewman} the bond rates as well 
      as the death rates are chosen at random which makes establishing uniqueness
      considerably more difficult, however allowing the environment to be chosen 
      ergodically by a sampling function adds no difficulty to the proof.

           For large  $L> 0$ we construct a bond-site percolation model on $\wt{ \mathbb{Z}}^{d+1} $ with a measure $\wt{P}_{p_L,q_L}$ coupled to the original percolation measure ${\bf Q_{\delta;\lambda}}$ on $\lxt$.   Let $\tau > 1$ and define time scale $T_L = e^{L^\tau}$. Define the map $J_L:\bbZ^d\times \bbR \to \bbZ^{d+1}$ by 
      \[J_L(x,t) =
         \left(\lfloor  (x_1 + L)/(2L+1) \rfloor,
         \cdots ,  (x_d + L)/(2L+1),
          \lfloor t/T_L \rfloor\right).  \]
   Let $\wt e_i$ be the standard basis of $\wt \bbZ^{d+1}$, with a $1$ at the $i^{th} $ position and zeros at other positions.  Sites in the bond-site percolation model are occupied with probability  
 \[ q_L = \exp\{\exp\{-\tfrac12 L^\tau\}\}. \]
           Given $\chi'> 0$, nearest neighbor bonds $\{\wt x, \wt y\}$ so that $\wt y = \wt x \pm \wt e_i$ for $i = 1,..,d$ are occupied with probability
           \[ p_L = 1 - \exp\{ - T_L e^{-cL} \},  \]
           for $c>0$ depending only on $d,\|h\|_\infty$ and $\lambda$.
         Nearest neighbor bonds $\{\wt x, \wt y\}$, so that $\wt y = \wt x \pm \wt e_{d+1}$, are occupied with probability one. 
 
   Recall, for $X \in \bbZ^d \times \bbR$ we write $C(X)$ for the connected cluster containing $X$. Similarly, we write $\wt C(0)$ for the connected cluster in $\wt Z^{d+1}$ containing the origin. For a cluster $C \subset \bbZ^d\times \bbR$, $|C|$ refers to the Lebesgue measure of the intervals contained in the cluster. For  $\wt C \subset \bbZ^{d+1}$, $|\wt C|$ is simply the counting measure.

   \begin{proposition}\label{cpl47}
        Let $(\bfTh, \bfT)$ be an environment process.
     Suppose there is  some $\chi > 0$, so that
     \be\label{ps47} \limsup_{ \epsilon \to 0 }   \frac{ \log|\log \psi_h(\epsilon)| }{\log K_2(\epsilon) }     >  \chi .     \ee
     Let $\tau>0 $ be chosen so that $\tau < \chi $. Let $\lambda> 0 $ and $K < \infty$, there is a finite $L > K$ so that the following holds for all initial conditions $\theta \in \bfTh$.
  \begin{itemize}
    \item[1.] There is $\ol x\in \bbZ^d$ so that $|\ol x | \leq L$ and  
    \[ \bfQ_{h,\theta;\lambda}( |C(\ol x,0)| > N T_L ) \geq \wt P_{p_L,q_L}(|\wt C(0)| > N). \]  
    
    \item[2.]
     For any  $(x,t),(y,s) \in \bbZ^d \times \bbR$, there are $\ol x,\ol y \in \bbZ^d$ so that $\|\ol x - x\|, \|\ol y - y\|\leq L$ and 
     \be \label{commlow2} \bfQ_{h,\theta;\lambda}((\ol x,t) \lra (\ol y,s) ) \geq  \wt P_{p_L,q_L}(J_L(x,t) \lra J_L(y,s) )   \ee
     \end{itemize}
   \end{proposition}
  In the bond-site percolation model, when $J_L(x,t) = J_L(y,s)$ we write $\{J_L(x,t) \lra J_L(y,s) \} $ only if the site $J_L(x,t) $ is occupied.

      \begin{proof}
        By (\ref{ps47}), there is a sequence of $\epsilon \to 0$, so that there is some $\theta_\epsilon \in h^{-1}(0)$ so that for every $\theta$ so that $r(\theta,\theta_\epsilon) < \epsilon$,
        \[  -\log h(\theta) \geq  -\log \psi_h(\epsilon) \geq K_2^{\chi}(\epsilon). \]
        By definition of $K_2$, for any $\theta\in \bfTh$ there is some $x $ so that $|x| \leq K_2(\epsilon)$, we have $r(\bfT^x \theta,\theta_\epsilon)$ therefore
        \[  h(\bfT^x\theta) \leq e^{ - K_2^{\chi}(\epsilon)} . \] 
        Choose $\epsilon$ small enough that $ K_2(\epsilon) \geq K$, let $L = K_2(\epsilon) $

          Formally, for each $\wt x\in \wt \bbZ^{d+1}$  associate to it the box in $\bbZ^d\times \bbR$ 	
          \[B(\wt x;L) = \Lambda_L((2L+1)\wt x_{1,..,d}) 
               \times [T_L \wt x_{d+1},T_L(\wt x_{d+1} + 1)]\]
           where $\wt x_{1,..,d} = (\wt x_1,..,\wt x_d) \in \bbZ^d $.
          Let 
          \[u_{\wt x} =
           \arg\min
            \{\delta(x): x\in \Lambda_L((2L+1)\wt x_{1,..,d}) \}.\]  
          Now define the interval
          \[ I_{\wt x} = u_{\wt x}\times[T_L\wt x_{d+1},T_L(\wt x_{d+1}+1)].\] 
          and we define the event of occupation at site $\wt x$
          \[  W_{\wt x} = \{  \tn{No deaths occur on the interval }I_{\wt x} \} \]
          For $\wt x,\wt y \in \tilde \bbZ^{d+1}$, the events $W_{\wt x}$ and $W_{\wt y}$ are independent if $\wt x \neq \wt y$.   For nearest neighbors $\wt x + \wt e_{d+1} =  \wt y $ we consider the bond  between $\tilde{x}$ and $\tilde{y}$  occupied with probability 1. For nearest neighbors $\wt x + \wt e_{i} =  \wt y $, for $i = 1,..,d$ we consider the bond  between $\wt{x}$ and $\wt{y}$  occupied in the event  
          \[ E(\wt x,\wt y;L) =  \{  I_{\wt x}  \lra  I_{\wt y} : \tn{ within } B(\wt x;L) \cup B(\wt y;L)\}. \]

           Let $(\wt x^{(1)},\wt y^{(1)}), ..,(\wt x^{(n)},\wt y^{(n)})$ be any collection of nearest neighbors in $\wt \bbZ^{d+1}$. By the FKG inequality (\ref{fkginfty})
   \[   \wt P ( (\wt x^{(j)},\wt y^{(j)}) \tn{ are occupied for } j=1,...,n  ) \geq \prod_{j =1,..,n} \wt P (  (\wt x^{(j)},\wt y^{(j)}) \tn{ is occupied}  ).\] 

          We now  estimate the occupation probabilities.
          Let $\wt{x},\wt{y}\in\tilde{\zee}^{d+1}$ 
          so that $\wt{x} + e_i = \wt{y}$ for some $i = 1,..d$. Split $B(\wt x;L) \cup B(\wt y,L)$ into $T$ similar \textquoteleft slices\textquoteright, ie,
          for $j \in \{1,\ldots,T\}$ let 
         \[ D_j(\wt{x},\wt{y};L) 
                 = B(\wt x;L) \cup B(\wt y,L)  \bigcap  \left(\zd\times (\wt{x}_{d+1}T + i-1, \wt{x}_{d+1}T + i)\right) \]
          then consider  $ E_i(\tilde{x},\tilde{y};L)$ the event that $u_{\wt x} $ communicates to $u_{\wt y} $ within $ D_j(\wt{x},\wt{y};L)$.
         First  observe that, for any points $\ol{x} \in \Lambda_L((2L+1)\wt x_{1,..,d})$ and $\ol{y} \in \Lambda_L((2L+1)\wt y_{1,..,d})$, there is a path connecting $\ol x$ to $\ol y$ within $\Lambda_L((2L+1)\wt x_{1,..,d}) \cup \Lambda_L((2L+1)\wt y_{1,..,d})$ of length less than $2d(2L+1) < 5dL$. 
          Consider the event $\wh E_i(\wt x,\wt y;L)$ that there are no deaths in $D_j(\wt{x},\wt{y};L) $ on each point in the path and and there is a bond in $D_j(\wt{x},\wt{y};L) $ for every step in the path.
           It is easy to see that $\wh E_i(\wt x,\wt y;L) \subset  E_i(\wt x,\wt y;L) $ so that
         \[\bfQ_{h,\theta;\lambda}(E_i(\wt{x},\wt{y};L))
              \geq 
            \left(1- e^{-\lambda}\right)^{5dL}\left(e^{-\|h\|_\infty}\right)^{5dL}
            \geq\exp\{ -cL \},\]
          for some $c<\infty$ depending only on $d, \|h\|_\infty$ and $\lambda$.
          Observe, the bond between $\tilde x$ and $\tilde y$ is occupied if $D_i(\tilde x,\tilde y)$ holds for some $i$ 
          therefore 
          \[   E^c(\wt x,\wt y;L) \subset \cap_{i=1}^T E_i^c(\wt x,\wt y;L) .\] 
           The events $E_i(\wt x,\wt y;L)$ are independent by construction. Therefore, using the FKG inequality (\ref{fkginfty})  we have   
          \[  \bfQ_{h,\theta;\lambda}
           (E(\wt x,\wt y;L))
            \geq  
            \left( 1 - \left(1 - e^{-cL}\right)^{T_L}\right)
           \geq 1 -  \exp\{- T_L e^{ - cL}\}. \] 
          On the other hand,   by assumption $\delta_{u_{\tilde x}} < \exp\{ - L^{\chi} \}$ so we have
          \[  \bfQ_{h,\theta;\lambda} (W_{\wt x})
             =  \exp \{-T  \delta(u_{\wt x}) \} 
             > \exp\left\{- \exp\{ L^\tau  -  L^{\chi}\} \right\} > \exp\{\exp\{-\tfrac12 L^\tau\}\}.  \] 
  This completes the first claim of the Proposition. Indeed, we simply consider the cluster  $C(\ol x,0)$ for $\ol x = u_{\wt 0}$ and compare it to the cluster $\wt C(0)$.
  We now show the final claim. Let $F_L:\bbZ^d\mapsto \bbZ^d $ be defined by letting $ \wh x = \arg\min \{\| (2L+1)y - x\|: y \in \bbZ^d \} $ and
   \[ F_L(x) = \arg\min \{ \delta(y): y\in \Lambda_L ((2L+1) \wh x )  \}, \]  
 so that $F_L$ maps to the site minimizing $\delta$ in the block `standard grid'.  
  Setting $\ol x = F_L(x)$ and $\ol y =  F_L(y)$ obtains (\ref{commlow2}). 
   \end{proof}

\begin{proof}[proof of Theorem \ref{hirec}.]
          By the first part of Proposition \ref{cpl47} there exists an infinite component if the bond-site percolation model has an infinite component. As $\chi > 1$, we can choose $\chi > \tau > 1$. For any $\lambda > 0$, we can choose $K$ sufficiently large so that $q_L$ and $p_L$ are arbitrarily close to 1. From a standard Peierl's argument, using $p_L$ and $q_L$ close enough to 1, there exists an infinite component with probability 1. Thus $h \in F_{strong}$ which completes the proof.
      \end{proof}

\subsection{Rotations on $\bbT$}\label{rotations}
 \begin{proof}[proof of Theorem \ref{rotationtheorem}] 
  We begin with part (1.).

     Let us review facts from the theory of continued fractions \cite{k64c}.
      Let $p_n/q_n$ be the sequence of approximants defined in (\ref{cfracrep})
      the sequence of denominators is defined as $q_{-2} = 0,q_{-1} = 1 $ and $q_n = a_n q_{n-1} + q_{n-2}$.
        Moreover, for any interval $I \subset \bbT$ so that $|I| > 1/q_n$, and any $\theta\in\bbT$, 
          there is some $k$ so that $1 \leq k \leq q_n + q_{n-1}$ and $\theta + k\omega \in I$.
    Thus, by definition, $ K_2( \tfrac1{2q_n}) \leq q_n + q_{n-1} < 2 q_n$.
         
         From the assumption on $h$, for sufficiently large $n$ and any $\theta$ so that $r (\theta, \theta_0 ) < \tfrac1{2q_n} $
         \[  \log|\log h(\theta)| > a |\log \frac{1}{2q_n}|  > a \log K_2(\frac{1}{2q_n} ) \]
          Thus $h$ satisfies (\ref{psilimit2}), this completes the proof of statement (1.).

  Now let us prove statement $(2.)$.   
   Again we recall a standard result in continued fraction theory. Let $r$ be the metric on the torus. For any $ n > 0$ and all  $ 0 < q \leq q_{n}$, 
   \[ |r( q 2\pi  \omega, 0)| > |r( q_n 2\pi  \omega,0)| > \frac{1}{2 q_{n+1} }. \]
   So that $K_1(\tfrac{1}{2q_{n+1}}) \geq q_n$, by the assumption on $h$, for large  $n > 0$ and $\theta$ so that $r(\theta, h^{-1}(0)) > \frac1{2q_{n+1}} $,
   \[ \log |\log h(\theta)| < a |\log \frac{1}{2q_{n+1}}| \leq a\log 2C_\omega + a(1 + \gamma ) \log q_n \]
   Thus, for large enough $n$,
   \[
   \frac{\log |\log h(\theta)|}{ \log K(\frac{1}{2 q_{n+1}})} <  a(1 + \gamma ). 
   \]
    Thus $h$ satisfies  (\ref{psilimit}), so for $  \frac{1  + \alpha (1 + \gamma)} {2} < \nu < 1$ and $m  > 0$, there is small enough $\lambda > 0$ so that the pair  $(h,\lambda)$ is $\bfC_{m,\nu}$-localized.
      
 \end{proof}

\bibliographystyle{plain}	
\bibliography{qimrefs}

\end{document}